\documentclass[12pt,a4paper]{article}

\usepackage{etex}
\usepackage[utf8]{inputenc}
\usepackage[T1]{fontenc}
\usepackage[american]{babel}

\usepackage{tikz}
\usetikzlibrary{boxes}
\usetikzlibrary{calc,arrows.meta}

\usepackage[hmargin=3.25cm,vmargin=2.5cm]{geometry}
\usepackage{setspace}
\spacing{1.15}
\usepackage{enumitem}
\setlist[enumerate,1]{leftmargin=*,label={\textnormal{(\roman*)}},topsep=0.5ex}
\usepackage{booktabs}

\usepackage{amssymb,amsmath}
\usepackage[swapvars]{moremath}
\newcommand{\dpunct}[1]{\text{#1}}
\newdelimcommand{clen}{\lVert}{\rVert}
\usepackage{mathtools}
\mathtoolsset{
  mathic=true,% Italics correction im Mathemodus
}
\usepackage{complex}
\renewrobustcmd{\pname}[1]{{\normalfont\scshape #1}}
\newcommand{\dif}{\mathop{\mathrm{d}}\!}
\newcommand{\od}[2]{\frac{\dif #1}{\dif #2}}

\usepackage{amsthm}
\theoremstyle{plain}
\newtheorem{theorem}{Theorem}
\newtheorem{lemma}[theorem]{Lemma}
\newtheorem{corollary}[theorem]{Corollary}
\theoremstyle{definition}
\newtheorem{definition}[theorem]{Definition}

\newcommand{\ptm}[3][a]{\fcall[#1]{\operatorname{t}_{#2}}{#3}} % Processing time of job #2 with #3 machines (malleable tasks)
\newcommand{\work}[3][a]{\fcall[#1]{\operatorname{w}_{#2}}{#3}} % Work of job #2 with #3 machines (malleable tasks)
\newcommand{\procnum}[3][a]{\fcall[#1]{\gamma_{#2}}{#3}} % Least number of machines for job #2 to finish in less than #3
\newcommand{\KP}[3][]{\fcall{\operatorname{KP}\ifstrequal{#1}{1}{^{(1)}}{}\ifstrequal{#1}{2}{^{(2)}}{}}{#2, #3}}
\newcommand{\KPOPT}[3][a]{\fcall[#1]{\operatorname{OPT_{KP}}}{#2, #3}}
\newcommand{\KPCOPT}[2][a]{\fcall[#1]{\operatorname{OPT}}{#2}}
\newcommand{\s}[2][1]{\fcall[#1]{\operatorname{s}}{#2}}  % Size of item
\newcommand{\pr}[2][1]{\fcall[#1]{\operatorname{p}}{#2}} % Profit of item

\usepackage{hyperref}
\usepackage{bookmark}
\usepackage[capitalize]{cleveref}

\let\chapter\section
\usepackage[ruled,vlined,linesnumbered]{algorithm2e}
\crefname{algocf}{Algorithm}{Algorithms}
\crefname{algocfline}{line}{lines}
\DontPrintSemicolon
\SetKwSty{textbf}
\SetFuncSty{texttt}
\SetArgSty{}
\SetDataSty{textsf}

\SetCommentSty{textit}
\SetKwComment{tpy}{\# }{}

\SetKw{KwBreak}{break}
\SetKwInOut{Input}{Input}

\title{Scheduling Monotone Moldable Jobs in Linear Time\footnote{%
  Research was in part supported by German Research Foundation (DFG) project JA 612/16-1.
}}
\author{Klaus Jansen \and Felix Land}

\begin{document}
\maketitle

\begin{abstract}
  A moldable job is a job that can be executed on an arbitrary number of processors,
  and whose processing time depends on the number of processors allotted to it.
  A moldable job is monotone if its work doesn't decrease
  for an increasing number of allotted processors.
  We consider the problem of scheduling monotone moldable jobs to minimize the makespan.
  
  We argue that for certain compact input encodings
  a polynomial algorithm has a running time polynomial in $n$ and $\log m$,
  where n is the number of jobs and m is the number of machines.
  We describe how monotony of jobs can be used to counteract the
  increased problem complexity that arises from compact encodings,
  and give tight bounds on the approximability of the problem with compact encoding:
  it is \cclass{NP}-hard to solve optimally, but admits a PTAS.
  
  The main focus of this work are efficient approximation algorithms.
  We describe different techniques to exploit the monotony of the jobs for better running times,
  and present a $(\frac{3}{2}+\epsilon)$-approximate algorithm
  whose running time is polynomial in $\log m$ and $\frac{1}{\epsilon}$,
  and only linear in the number~$n$ of jobs.
\end{abstract}

\section{Introduction}

In classical scheduling models, the input consists of a description of the available processors,
and a set of jobs with associated processing times.
Each processor can process one job at any point in time.
Additional constraints may be part of the model.
One way to model complex, parallelizable tasks are moldable\footnote{Some authors use the term malleable.} jobs,
which have a variable parallelizability~\cite{du89}.
Formally, we are given a set~$J$ consisting of $n$~jobs and a number~$m$ of processors.
The processing time $\ptm{j}{1}$ on one processor is given for each job~$j$,
as well as the speedup~$\fcall{\mathrm{s}_j}{k}$ that is achieved when executing it on $k > 1$ processors.
The processing time on $k$ processors then is given as~$\ptm{j}{k} = \frac{\ptm{j}{1}}{\fcall{\mathrm{s}_j}{k}}$.
The goal is to produce a schedule that assigns for each job a starting time and a number of allotted processors
such that the \emph{makespan}, i.e.~the completion time of the last job, is minimized.

Without restriction, we assume that the speedup is non-decreasing,
or equivalently, the processing time is non-increasing in the number of processors.
A job is called monotone if its work function $\work{j}{k} = k \times \ptm{j}{k}$ is non-decreasing.
This is a reasonable assumption, since an increased number of processors requires more communication.
Monotony helps when designing algorithms~\cite{belkhale90,mounie99,mounie07}.
Sometimes even stronger\footnote{For a proof that concave speedup functions imply monotony, see~\cite{jansen12}.}
assumptions are made, e.g.~that the speedup functions are concave~\cite{blazewicz06,sanders11,jansen12}.

Since the problems considered here are \cclass{NP}-hard,
we will discuss approximation algorithms.
An algorithm for a minimization problem is $c$-approximate
if it produces a solution of value at most $c \OPT*{I}$ for each instance~$I$.
The number~$c \geq 1$ is called its approximation guarantee.

We pay extra attention to the encoding length~$\clen{I}$ in dependence of the number~$m$ of processors.
The running time of most algorithms is polynomial in~$m$~\cite{belkhale90,turek92,jansen03,jansen06a,mounie07,jansen10}.
Many authors expect that the values $\ptm{j}{k}$, $k \in \setrange{1}{m}$ are explicitly given as a list,
such that $m = \Landau{<=}{\clen{I}}$.
Under this assumption, these algorithms' running time is polynomial in the input size.
On the other hand, more compact encodings are conceivable in many cases.
Sometimes it is assumed that the processing time function is of a special form,
e.g.~linear~\cite{grigoriev06} or a power function~\cite{makarychev14},
which can be described with a constant number of coefficients.
Since the number of processors is encoded in $\log m$ bits,
the aforementioned algorithms can have a running time that is exponential in the input length
when compact enconding is used.

It is our main goal to develop fully polynomial algorithms for instances with compact input encoding,
i.e.~algorithms with a running time polynomial in~$\log m$.
Such algorithms will outperform algorithms whose running time is polynomial in $m$
for large values of $m$ (super-polynomial in the input size).
Only few known algorithms are fully polynomial in this sense~\cite{mounie99,sanders11,jansen13a}.
Since we do not want to stipulate a certain form of speedup functions,
we assume that the running times~$\ptm{j}{k}$ can be accessed via some oracle in constant time.

\paragraph*{Previous Results}

It is known that finding an exact solution without monotony is \cclass{NP}-hard~\cite{du89}.
If the jobs are monotone,
it is only known that finding an exact solution is weakly \cclass{NP}-hard~\cite{jansen13c}.
Both results even hold for a constant number of processors.
As a consequence, they also hold with compact input encoding.
The problem complexity actually depends on the used input encoding.
It is known that there is no polynomial time algorithm for scheduling of parallel jobs
with approximation guarantee less than $\frac{3}{2}$,
unless $\cclass{P} = \cclass{NP}$.
This can be deduced from a reduction from the partition problem to scheduling of parallel jobs~\cite{drozdowski95}.
By setting
$\ptm{j}{k} = t_j$ if $k \geq \mathrm{size}(j)$ and $\ptm{j}{k} = \infty$ otherwise,
we can extend the result to scheduling of moldable jobs,
although the resulting work functions are not monotone.
However, this reduction is polynomial only if we use a compact encoding for the resulting instance.
Furthermore, if we allow algorithms to be polynomial in~$m$,
the produced instances can be optimally solved:
since the reduction is one-to-one, we can go back to the original partition instance,
solve it via dynamic programming in time~$\Landau{<=}{nm}$,
and convert the solution back to the scheduling setting.
Indeed, without compact encoding, the problem admits a PTAS~\cite{jansen10}.

Considering approximate algorithms,
Belkhale and Banerjee~\cite{belkhale90} found a $2$-approximate algorithm for scheduling monotone moldable tasks.
This approximation guarantee was later matched without monotony by an algorithm due to Turek, Wolf, and Yu~\cite{turek92}.
The running time was later improved by Ludwig and Tiwari~\cite{ludwig94}.
Their algorithm for the case of monotone jobs was the first to achieve a running time polynomial in $\log m$, namely $\Landau{<=}{n \log^2 m}$.
Mounié, Rapine, and Trystram improved the approximation guarantee with monotony to~$\sqrt{3} + \epsilon \approx 1.73$,
with arbitrarily small~$\epsilon > 0$,
also with polylogarithmic dependence on~$m$.
They later presented a $(\frac{3}{2} + \epsilon)$-approximate algorithm with running time~$\Landau{<=}{nm \log \frac{1}{\epsilon}}$~\cite{mounie04,mounie07}.
A PTAS with running time polynomial in~$m$ was subsequently developed that does not require monotony~\cite{jansen10}.
Finally, a $(\frac{3}{2} + \epsilon)$-approximate algorithm with polylogarithmic dependence on~$m$ that also does not assume monotone jobs was developed by Jansen~\cite{jansen13a}.

\paragraph*{Our Contribution}

We improve the understanding of scheduling monotone moldable jobs in several ways.
In \cref{sec:hardness} we resolve the complexity of the considered problem.
\begin{theorem}
  \label{thm:np-complete}
  It is \cclass{NP}-complete to decide whether a set of monotone jobs can be scheduled with a given makespan.
\end{theorem}
We proceed to describe an extremely efficient FPTAS for the case that the number of machines is large enough in \cref{sec:fptas}.
\begin{theorem}
  \label{thm:fptas}
  There is an FPTAS  for the case that $m \geq 8\frac{n}{\epsilon}$
  with a running time of $\Landau{<=}{n \log^2 m (\log m + \log \frac{1}{\epsilon})}$.
\end{theorem}
In combination with the PTAS by Jansen and Thöle~\cite{jansen10},
this yields a PTAS for scheduling of monotone moldable jobs
with compact encoding of running times.
The algorithm by Jansen~\cite{jansen13a} achieves the same approximation guarantee
in the more general case without monotony,
but has a significantly worse running time.
In particular, our new algorithm's running times are polynomial in $\frac{1}{\epsilon}$,
while Jansen's algorithm is doubly exponential in $\frac{1}{\epsilon}$.

In \cref{sec:approximation} we first describe the $(\frac{3}{2} + \epsilon)$-approximate algorithm
due to Mounié, Rapine, and Trystram~\cite{mounie07},
and improve its running time to fully polynomial.
We further present techniques to gradually reduce the dependence of the running time on the number~$n$ of jobs.
\begin{theorem}
  \label{thm:dual-approx}
  For each $T$ given in Table~\ref{tbl:running-times}, there is a $(\frac{3}{2} + \epsilon)$-approximate algorithm
  with running time~$\Landau{<=}{n \log^2 m + \log \frac{1}{\epsilon} T(n, m, \epsilon)}$.
  \begin{table}
    \caption{Running times of our $(\frac{3}{2} + \epsilon)$-dual algorithms.}
    \label{tbl:running-times}
    \centering
    \begin{tabular}{l l}
      \toprule
        Algorithm & $T(n, m, \epsilon)$ \\
      \midrule
        Section~\ref{sec:reducing-kp-problems} &
          $\Landau{<=}{n (\log m + n \log \epsilon m)}$ \\
        Section~\ref{sec:bounded-kp} &
          $\Landau[2]{<=}{n \parens[2]{\frac{1}{\epsilon^2} \log m \parens[2]{\frac{\log m}{\epsilon} + \log^3 (\epsilon m)} + \log n}}$ \\
        Section~\ref{sec:linear} &
          $\Landau[2]{<=}{n \frac{1}{\epsilon^2} \log m \parens[2]{\frac{\log m}{\epsilon} + \log^3 (\epsilon m)}}$ \\
      \bottomrule
    \end{tabular}
  \end{table}
\end{theorem}

We make repeated use of a technique we call compression.
It reduces the number of processors used by a job in exchange for a bounded increase in the running time.
Compression allows us to approximate processor numbers for jobs that are allotted to a large number of processors.
This enables the use of various rounding techniques.
The intermediate solution then may use more than~$m$ processors,
before the jobs are finally compressed such that they require at most~$m$ processors.
This is similar to models with resource augmentation (see e.g.~\cite{chekuri04}),
except that we can use additional processors only for jobs that are allotted to a large number of processors.

\section{\cclass{NP}-Completeness of Monotone Moldable Job Scheduling}
\label{sec:hardness}

In this section we discuss Theorem~\ref{thm:np-complete}.
To be precise, we consider the problem of deciding
whether a given instance of scheduling monotone moldable jobs can be scheduled with makespan at most~$d$,
where $d$ is also part of the input.

\begin{proof}[Proof of Theorem~\ref{thm:np-complete}]
We first argue that the problem is in \cclass{NP}
by giving a nondeterministic polynomial procedure for solving the problem:
first, guess the number of processors allotted to each job in an optimal schedule.
Since these numbers are at most~$m$, we can guess them one bit at a time in $n \log m$ steps.
Afterwards, we guess the order in which the jobs start.
This is a list of $n$ numbers in $\setrange{1}{n}$,
which can be encoded in $n \log n$ bits.
Again, we can guess this encoding in $n \log n$ steps.
We now use list scheduling to schedule the jobs in this order
while respecting the previously guessed processor counts.
This procedure is clearly possible in polynomial time and uses $n (\log m + \log n)$ (binary) guessing steps.

\begin{figure}[b]
  \centering
  \begin{tikzpicture}[x=0.75cm,y=0.4cm]
    \tikzset{container/.style={fill=gray!10}}

    \draw[|->,>=latex] (-1.5,0) -- (-1.5,6.5) node[midway,sloped,above] {\small Time};

	  \pnode[name=S,draw=none,fill=none]{width=8,height=6}{}
	  \drawschedule{S}
 	  \draw[|-|] (S.south west) -- (S.north west) node[pos=0,left] {\footnotesize $0$} node[pos=1,left] {\footnotesize $nB$};
    \draw[<->,>=latex] (0,-1) -- node[below] {$m$} (8,-1);
 	  
 	  \newcommand{\joblist}[5]{%
  	  \pnode[name=#11,container]{align left=#2,align bottom=(S.south),width=1,height=#3}{}
	    \pnode[name=#12,container]{above of=#11,height=#4}{}
  	  \pnode[name=#13,container]{above of=#12,height=#5}{}
  	  \pnode[name=#14,container]{above of=#13,align top=(S.north)}{}
  	}
  	
  	\joblist{A}{(S.west)}{2.0}{1.3}{1.7}
  	\joblist{B}{(A1.east)}{1.4}{1.5}{1.8}
  	\joblist{C}{(B1.east)}{1.7}{1.6}{1.2}
  	\coordinate (E1) at (5,0);
  	\node at (4,3) {$\cdots$};
  	\joblist{F}{(E1.east)}{1.3}{1.2}{1.7}
  	\joblist{G}{(F1.east)}{1.9}{1.5}{1.3}
  	\joblist{H}{(G1.east)}{1.6}{1.3}{1.6}
	\end{tikzpicture}
  \caption{Structure of a schedule with makespan~$nB$}
  \label{fig:optimal-schedule}
\end{figure}
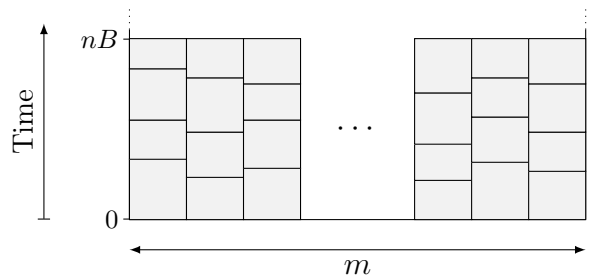

We give a reduction from \pname{$4$-Partition} to prove that our scheduling problem is strongly \cclass{NP}-hard.
Recall that an instance of \pname{$4$-Partition} contains a set~$A = \setirange{a}{1}{4n}$ of natural numbers and a number~$B$,
and remains \cclass{NP}-hard even when all numbers are strictly between $\frac{B}{5}$ and~$\frac{B}{3}$~\cite{garey79}.
We construct an instance of the scheduling problem as follows.
First, we assume that $\sum_{i = 1}^{4n} a_i = nB$,
otherwise we output a trivial no-instance.
Next, we scale the numbers such that~$a_i \geq 2$ for each~$i \in \setrange{1}{4n}$.
The number of machines will be~$m = n$.
Now we create a job~$j_i$ for each number~$a_i$, which has processing time~$\ptm{j_i}{k} = m a_i - k + 1$.
Note that these functions are monotonically decreasing.
Also~$m a_i \geq 2m > 2k$ for each~$k < m$ and therefore the work satisfies
\begin{equation}
  \label{eq:reduction-monotone}
  \begin{aligned}
    \work{j_i}{k+1} &= (k+1) \times \ptm{j_i}{k+1} \\
    &= (k+1)(m a_i + 1) - (k+1)^2 \\
    &= k (m a_i + 1) - k^2 + (m a_i + 1) - 2k - 1 \\
    &= \work{j_i}{k} + m a_i - 2k
    > \work{j_i}{k} \dpunct{,}
  \end{aligned}
\end{equation}
i.e~the jobs are strictly monotone.
The target makespan is~$d = nB$.
It remains to show that this is indeed a reduction,
i.e.~that a schedule with makespan~$nB$ exists if and only if the instance of \pname{$4$-Partition} is a yes-instance.

First assume that there is a schedule with makespan~$d$.
The total work of all jobs is at least~$\sum_{i = 1}^{4n} \work{j_i}{1} = \sum_{i = 1}^{4n} m a_i = m d$.
Due to the strict monotony, our schedule must allot exactly one processor to each job,
and all machines have load~$d$,
see \cref{fig:optimal-schedule} for an example of such a schedule.
The numbers corresponding to the jobs on one machine sum up to~$B$,
and because they are strictly between $\frac{B}{5}$ and~$\frac{B}{3}$,
there are exactly four such numbers.
Therefore, there is a solution to the instance of \pname{$4$-Partition}.

On the other hand, if the instance of \pname{$4$-Partition} is a yes-instance,
a schedule as depicted in \cref{fig:optimal-schedule} is easily constructed from a solution.
\end{proof}

This shows that scheduling monotone jobs is \cclass{NP}-hard in the strong sense.
Furthermore, using a complexity result for \pname{$4$-Partition}~\cite{jansen13b},
there is no algorithm that solves this problem exactly in time~$2^{\Landau{<}{n}} \times \clen{I}^{\Landau{<=}{1}}$,
unless the Exponential Time Hypothesis fails.

\section{An FPTAS for Large Machine Counts}
\label{sec:fptas}

In this section, we present a fully polynomial approximation scheme (FPTAS)
for the case that $m \geq 8\frac{n}{\epsilon}$,
as stated in Theorem~\ref{thm:fptas}.
An FPTAS finds a $(1 + \epsilon)$-approximate solution
in time polynomial in the input length and $\frac{1}{\epsilon}$
for each $\epsilon \in (0, 1]$.
The case where $m$ is much larger than $n$ is the most interesting case
in our setting with compact input encoding,
because otherwise $m$ is polynomial in the input.
Furthermore, this allows us to focus on the case $m < 8\frac{n}{\epsilon}$
in the following chapters.

The algorithm itself is a \emph{dual approximate algorithm}~\cite{hochbaum87}.
A $c$-dual approximate algorithm accepts a number $d$ in addition to the instance as input.
It will output a solution with makespan at most~$cd$,
provided that a solution with makespan~$d$ exists.
Otherwise, it may reject the instance.
It is well known that a $c$-approximate dual algorithm with running time~$T(n, m)$
can be turned into a $(c + \epsilon)$-approximate algorithm
with running time~$\Landau{<=}{T'(n, m) + \log \frac{1}{\epsilon} \times T(n, m)}$,
where $T'(n, m)$ is the running time of an \emph{estimation} algorithm with arbitrary but constant estimation ratio.

An estimation algorithm with estimation ratio~$\rho$ computes a value~$\omega$
that estimates the minimum makespan within a factor of~$\rho$,
i.e.~$\omega \leq \OPT \leq \rho \omega$.
Here, we use an algorithm due to Ludwig and Tiwari~\cite{ludwig94}
with running time~$T'(n, m) = \Landau{<=}{n \log^2 m}$.
Although they do not explicitly state this,
their algorithm can be trivially turned into one with estimation ratio~$2$:

Their algorithm computes an allotment $a$
which allots to each job~$j \in J$ a number~$a_j$ of processors,
and this allotment minimizes the value
\begin{equation}
  \omega = \min \parens[4]{\frac{1}{m} \sum_{j \in J} \work{j}{a_j},
    \max_{j \in J} \ptm{j}{a_j}}
\end{equation}
among all allotments.
Therefore $\omega \leq \OPT$.
On the other hand, the list scheduling algorithm,
applied to the instance with the fixed allotment~$a$,
produces a schedule of makespan at most~$2 \omega$~\cite{garey75},
so $\OPT \leq 2 \omega$.

For our algorithm, we specify $c = 1 + \epsilon$,
resulting in a $1 + 2 \epsilon$ approximation ratio.
Our algorithms will frequently schedule jobs using the least number of processors
such that its processing time is below a threshold $t$.
Therefore let~$\procnum{j}{t} = \min \setst{p \in \natupto{m}}{\ptm{j}{p} \leq t}$,
see also the work of Mounié, Rapine, and Trystram~\cite{mounie07}.
Note that $\procnum{j}{t}$ can be found in time~$\Landau{<=}{\log m}$ by binary search.

The algorithm is extremely simple:
allot $\procnum{j}{(1+\epsilon) d}$ processors to each job~$j$
and schedule them simultaneously.
If this schedule requires more than $m$ machines, reject.

Clearly, the running time for the dual approximate algorithm is $\Landau{<=}{n \log m}$.
The final algorithm therefore requires
$\Landau{<=}{n \log m (\log m + \log \frac{1}{\epsilon})}$ time.

\subsection{Analysis}

It remains to show that the algorithm is indeed $(1+\epsilon)$-dual approximate.
The produced schedule clearly has makespan at most $(1+\epsilon) d$
by the definition of $\procnum{j}{(1+\epsilon) d}$.
To prove that the algorithm only rejects if there is no schedule with makespan~$d$,
we will argue that the produced schedule requires at most~$m$ processors,
i.e.~$\sum_{j \in J} \procnum{j}{(1+\epsilon) d} \leq m$,
provided that $d \geq d^*$,
where $d^*$~is the optimal makespan.

To this end, we consider the same algorithm with a more complex allotment rule,
which uses \emph{compression} of jobs,
our main technique for exploiting the monotony of the work functions.
Compression reduces the number of processors allotted to a job
in exchange for a bounded increase of its processing time.

\begin{lemma}
  \label{lemma:compression}
  If $j$ is a job that uses $b \geq \frac{1}{\rho}$~machines in some schedule,
  where~$\rho \in (0, 1/4]$,
  then we can free $\ceil{b\rho}$~machines and the schedule length increases
  by at most~$4\rho \ptm{j}{b}$.
\end{lemma}
We call the value $\rho$ the \emph{compression factor}.

\begin{proof}
  Formally, the statement of the lemma is
  $\ptm{j}{\floor{b(1-\rho)}} \leq (1 + 4\rho) \ptm{j}{b}$.
  For the proof we
  set~$b' = \ceil{b(1-2\rho)} \leq b$.
  Since~$b \geq \frac{1}{\rho}$ we have~$b\rho \geq 1$
  and thus~$b' \leq \floor{b(1-\rho)}$.
  This implies~$\ptm{j}{\floor{b(1-\rho)}} \leq \ptm{j}{b'}$
  Because our jobs are monotonic we have
  \begin{equation}
    \ptm{j}{b'} \times b' = \work{j}{b'} \leq \work{j}{b} = \ptm{j}{b} \times b \dpunct{.}
  \end{equation}
  Hence (and because~$1-2\rho \geq 1/2$) it follows that
  \begin{equation}
    \begin{aligned}
      \ptm{j}{b'} &\leq \ptm{j}{b} \times \frac{b}{b'} \leq \ptm{j}{b} \times \frac{b}{b(1-2\rho)} \\
      &= \ptm{j}{b} \times \parens[a]{\frac{1 - 2\rho}{1 - 2\rho} + \frac{2\rho}{1-2\rho}}  \\
      &\leq (1 + 4\rho) \ptm{j}{b}
    \end{aligned}
  \end{equation}
  and the lemma follows.
\end{proof}

Our second allotment rule has two steps.
\begin{enumerate}
  \item Allot $a_j = \procnum{j}{d}$ processors to each job~$j$.
  \item Compress each job that is allotted to at least $\frac{4}{\epsilon}$ processors
    with a factor of $\rho = \frac{\epsilon}{4}$.
\end{enumerate}
Note that $\rho \leq \frac{1}{4}$ because we assumed $\epsilon \leq 1$.
According to Lemma~\ref{lemma:compression},
each job has processing time at most $(1 + \epsilon) d$ with this allotment rule.

We claim that the resulting schedule requires at most $m$~processors.
Assume that this is not the case after the first step, i.e.~$\sum_{j \in J} a_j > m$,
otherwise the statement clearly holds.
Then the number of required processors is still bounded:

\begin{lemma}
  \label{lemma:first-step}
  If $d \geq d^*$ we have $\sum_{j \in J} a_j < m + n$.
\end{lemma}

\begin{proof}
  Let $J' = \setst{j \in J}{a_j = 1}$ and assume the statement holds
  if we remove the jobs in~$J'$ and their allotted machines,
  i.e.~$\sum_{j \in J \setminus J'} a_j < (m-\card{J'}) + (n-\card{J'})$.
  Then $\sum_{j \in J} a_j = \sum_{j \in J \setminus J'} a_j + \card{J'}
  < m + n - \card{J'} \leq m + n$.
  
  It is therefore sufficient to show the statement for jobs with $a_j > 1$.
  Assume that $d \geq d^*$.
  Then there is a schedule with makespan at most~$d$.
  For each job~$j$ let $a_j^*$ be the number of allotted processors in this schedule.
  Then $a_j \leq a_j^*$.
  Using the monotony of the work function, we have
  \begin{equation}
    \begin{aligned}
      \parens[4]{\sum_{j \in J} (a_j - 1)} d &= \sum_{j \in J} (a_j - 1)d  \\
      &< \sum_{j \in J} (a_j - 1) \ptm{j}{a_j -1} \\
      &= \sum_{j \in J} \work{j}{a_j - 1} \\
      &\leq \sum_{j \in J} \work{j}{a_j} 
       \leq \sum_{j \in J} \work{j}{a_j^*}
       \leq m d \dpunct{.}
    \end{aligned}
  \end{equation}
  Therefore $\parens[2]{\sum_{j \in J} a_j} - n = \sum_{j \in J} (a_j - 1) < m$,
  proving the lemma.
\end{proof}

Now partition the jobs into narrow and wide jobs, $J = J_N \cup J_W$.
The wide jobs are those that are compressed in the second step,
i.e.~$J_W = \setst{j \in J}{\procnum{j}{d} \geq \frac{1}{\rho}}$ and $J_N = J \setminus J_W$.
Let $\alpha = \sum_{j \in J_W} a_j$ and $\beta = \sum_{j \in J_N} a_j$.

In the second step, at least $\rho \alpha$ processors are freed.
By definition of the narrow jobs we have $\beta \leq n \frac{1}{\rho} = 4\frac{n}{\epsilon} \leq \frac{m}{2}$.
Since we assumed that $\alpha + \beta > m$ we have $\alpha > \frac{m}{2}$.
Therefore $\rho \alpha > \rho \frac{m}{2} \geq n$.
According to Lemma~\ref{lemma:first-step} we also have $\alpha + \beta \leq m + n$.
It follows that $(1-\rho)\alpha + \beta \leq m$, proving our claim.

To prove the claim about the first allotment rule,
we note that it cannot use more processors for any job, because it picks the minimum number of allotted processors when we target a makespan of $(1+\epsilon)d$.
Therefore, our algorithm is $(1+\epsilon)$-dual approximate,
proving Theorem~\ref{thm:fptas}.

\subsection{A PTAS for the General Case}

For the general case we can still achieve a PTAS.
When $m \geq 8 \frac{n}{\epsilon}$, simply use the previously described algorithm.
Otherwise apply the algorithm by Jansen and Thöle~\cite{jansen10}.
It is $(1 + \epsilon)$-approximate and has a running time polynomial in $n$ and $m$ (but exponential in $\frac{1}{\epsilon}$).
Since we use this algorithm only in the case that $m < 8 \frac{n}{\epsilon}$,
the running time is polynomial even when a compact input encoding is used for the processing times.
\section{A Linear $(\frac{3}{2} + \epsilon)$-Approximation}
\label{sec:approximation}

Unfortunately, the running time of the PTAS used in the last section is rather prohibitive.
We therefore want to develop more efficient algorithms.
The algorithms we present in this section
are modifications of the algorithm due to Mounié, Rapine, and Trystram~\cite{mounie07}
which has running time~$\Landau{<=}{nm}$.
We aim for a fully polynomial algorithm
with a running time that depends only linearly on the number~$n$ of jobs.

We achieve this goal with several modification to the original algorithm.
Before explaining our enhancements, we describe the original algorithm
by Mounié, Rapine, and Trystram.
In section~\ref{sec:fully-polynomial}
we show how to modify the algorithm
such that the running time is logarithmic in the number~$m$ of of machines.
With these modifications, the dependence of the running time on~$n$ actually increases.
We can however improve upon the general idea in section~\ref{sec:bounded-kp}.

\subsection{The Original Algorithm}

The algorithm is $\frac{3}{2}$-dual approximate.
Let $d$ be the target makespan.
The algorithm is based on the observation that all jobs
with a running time larger than $\frac{d}{2}$ in a feasible schedule of makespan~$d$ can be executed in parallel.

\paragraph*{Removing the Small Jobs}

We partition the jobs into small and big jobs, $J = J_S(d) \cup J_B(d)$.
Small jobs are jobs that complete in time $\frac{d}{2}$ on one machine,
i.e.~jobs~$j \in J$ with $\ptm{j}{1} \leq \frac{d}{2}$.
All other jobs will be denoted as big jobs.
We remove the small jobs from the instance,
they will be re-added in a greedy manner at the end of the algorithm.
Let $W_S(d) = \sum_{j \in J_S(d)} \ptm{j}{1}$ be the total work of these jobs.

\paragraph*{Finding a Preliminary Schedule}
\label{sec:knapsack}

In this step, a schedule for the big jobs is constructed
by placing them in two shelves,
shelf~$S_1$ with processing time~$d$ and shelf~$S_2$ with processing time~$\frac{d}{2}$.
The shelves are scheduled after each other for a total makespan of~$\frac{3}{2}d$.
Since a feasible schedule of this type may not exist,
we allow the second shelf to use more than~$m$ processors, see \cref{fig:2-shelf-schedule}.
We call such a schedule a two-shelf-schedule.
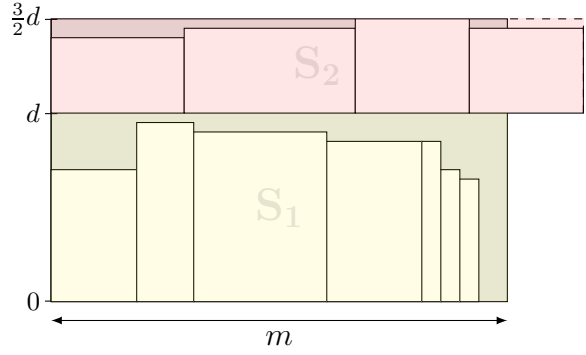
\begin{figure}
  \centering
  \begin{tikzpicture}[x=0.25cm,y=0.25cm]
    \draw[fill=gray!20] (0,0) rectangle (24,15);

 	  \draw[|-|] (0,0) -- (0,10) node[pos=0,left] {\footnotesize $0$} node[pos=1,left] {\footnotesize $d$};
 	  \draw[-|] (0,10) -- (0,15) node[pos=1,left] {\footnotesize $\frac{3}{2}d$};
    \draw[<->,>=latex] (0,-1) -- node[below] {$m$} (24,-1);
    
    \namedpnode[fill=white]{x-min=0,y-min=0,width=4.5,height=7}{J1}{}
    \namedpnode[fill=white]{right of=J1,y-min=0,width=3,height=9.5}{J2}{}
    \namedpnode[fill=white]{right of=J2,y-min=0,width=7,height=9}{J3}{}
    \namedpnode[fill=white]{right of=J3,y-min=0,width=5,height=8.5}{J4}{}
    \namedpnode[fill=white]{right of=J4,y-min=0,width=1,height=8.5}{J5}{}
    \namedpnode[fill=white]{right of=J5,y-min=0,width=1,height=7}{J6}{}
    \namedpnode[fill=white]{right of=J6,y-min=0,width=1,height=6.5}{J7}{}

    \namedpnode[fill=white]{x-min=0,y-min=10,width=7,height=4}{J8}{}
    \namedpnode[fill=white]{right of=J8,y-min=10,width=9,height=4.5}{J9}{}
    \namedpnode[fill=white]{right of=J9,y-min=10,width=6,height=5}{J10}{}
    \namedpnode[fill=white]{right of=J10,y-min=10,width=6,height=4.5}{J11}{}

    \draw[dashed] (J11.south east) |- (24,15);
    \pnode[draw=none,fill=yellow,fill opacity=0.1,text opacity=0.75]{x-min=0,y-min=0,x-max=24,y-max=10}{\Large $\mathbf{S_1}$}
    \pnode[draw=none,fill=red,fill opacity=0.1,text opacity=0.75]{x-min=0,y-min=10,align right=(J11.east),y-max=15}{\Large $\mathbf{S_2}$}
    
  \end{tikzpicture}
  \caption{Example of an infeasible two-shelf-schedule}
  \label{fig:2-shelf-schedule}
\end{figure}
A two-shelf-schedule can be found by solving a knapsack problem
where shelf~$S_1$ uses at most~$m$ processors and the profit of each job
is the amount of work saved when it is scheduled in shelf~$S_1$ instead of shelf~$S_2$,
i.e.~$v_j(d) = \work{j}{\procnum{j}{\frac{d}{2}}} - \work{j}{\procnum{j}{d}}$.
Note that the monotony implies that $v_j(d) \geq 0$.
If $\procnum{j}{d}$ is undefined for any job~$j$, i.e.~$\ptm{j}{m} > d$,
then we can safely reject $d$.
Each job~$j \in J_B(d)$ for which $\procnum{j}{\frac{d}{2}}$ is undefined,
i.e.~$\ptm{j}{m} > \frac{d}{2}$,
must be scheduled in~$S_1$.
Those jobs can be easily handled by removing them from the knapsack problem
and reducing the capacity accordingly.
In order to keep the notation simple,
we will however assume that no such jobs exist.
We denote the knapsack problem
\begin{equation}
  \begin{gathered}
    \max_{J' \subseteq J_B(d)} \sum_{j \in J'} v_j(d) \qquad
    \text{s.t.}\quad \sum_{j \in J'} \procnum{j}{d} \leq m
  \end{gathered}
\end{equation}
by $\KP{J_B(d)}{m,d}$ and the profit of an optimal solution by~$\KPOPT{J_B(d)}{m,d}$.
Solving the knapsack problem requires time $\Landau{<=}{nm}$
with a standard dynamic programming approach.
Calculating the required values of $\gamma_j$ can be done beforehand in time~$\Landau{<=}{n \log m}$,

Let $J'$ be a solution to this knapsack problem.
Then we can create a two-shelf-schedule
by placing the jobs $J'$ in $S_1$ and the jobs $J_B(d) \setminus J'$ in~$S_2$.
The work of this schedule is
\begin{equation}
  \label{eq:two-shelf-work}
  \begin{aligned}
    W(J',d) &= \sum_{j \in J'} \work{j}{\procnum{j}{d}} +
\sum_{j \in J_B(d) \setminus J'} \work{j}{\procnum{j}{\textstyle\frac{d}{2}}} \\
    &= \sum_{j \in J_B(d)} \work{j}{\procnum{j}{\textstyle\frac{d}{2}}} - \sum_{j \in J'} v_j(d) \dpunct{.}
  \end{aligned}
\end{equation}
At this point, we will reject $d$ if $W(J',d)$ is found to be larger than $md - W_S(d)$.
The correctness of this step is shown in the following lemma.

\begin{lemma}[\cite{mounie07}]
  \label{lemma:two-shelf-schedule}
  If there is a schedule for all jobs with makespan~$d$,
  then there is a solution $J'$ to $\KP{J_B(d)}{m,d}$
  with $W(J', d) \leq md - W_S(d)$.
\end{lemma}

\begin{proof}
  We only consider the big jobs in the feasible schedule,
  and their total work~$W_B$ is bounded from above by~$md - W_S(d)$,
  because $W_S(d)$ is a lower bound on the work of the small jobs.
  The jobs with processing time larger than~$\frac{d}{2}$
  induce a feasible solution~$J'$ to the knapsack problem.
  Then
  \begin{equation}
    \begin{aligned}
      W(J', d) &= \sum_{j \in J_B(d)} \work{j}{\procnum{j}{\textstyle\frac{d}{2}}} - \sum_{j \in J'} v_j \\
      &= \sum_{j \in J'} \work{j}{\procnum{j}{d}} 
       + \sum_{j \in J_B(d) \setminus J'}  \work{j}{\procnum{j}{\textstyle\frac{d}{2}}} \\
      &\leq W_B 
      \leq md - W_S(d) \dpunct{,}
    \end{aligned}
  \end{equation}
  so $J'$ is the claimed solution.
\end{proof}

Otherwise, the following step will transform the two-shelf-schedule
into a feasible schedule for all jobs.

\begin{lemma}[\cite{mounie07}]
  \label{lemma:make-feasible}
  Let $J' \subseteq J_B(d)$ be a solution of $\KP{J_B(d)}{m,d}$
  with $W(J', d) \leq md - W_S(d)$.
  Then we can find a schedule for all jobs in~$J$
  with makespan $\frac{3}{2}d$ in time~$\Landau{<=}{n \log n}$.
\end{lemma}

\subsubsection{Proof of Lemma~\ref{lemma:make-feasible}}

\paragraph{Obtaining a Feasible Schedule}
\label{sec:transformation-rules}

We exhaustively apply three transformation rules
that move some jobs to a new shelf~$S_0$.
The jobs in $S_0$ will be scheduled concurrent to $S_1$ and $S_2$
in $\frac{3}{2} d$ time units. creating a three-shelf-schedule.
For an example, see \cref{fig:3-shelf-schedule}.
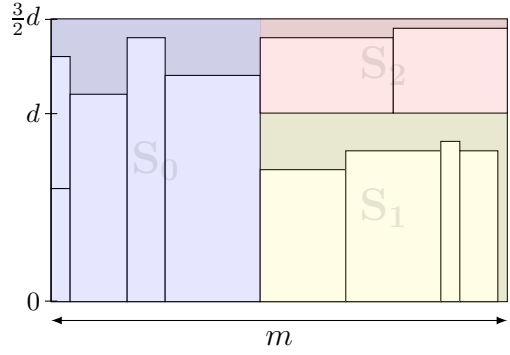
\begin{figure}
  \centering
  \begin{tikzpicture}[x=0.25cm,y=0.25cm]
    \draw[fill=gray!20] (0,0) rectangle (24,15);
    \draw[<->,>=latex] (0,-1) -- node[below] {$m$} (24,-1);
    
    \namedpnode[fill=white]{x-min=0,y-min=0,width=1,height=6}{J7}{}
    \namedpnode[fill=white]{x-min=0,above of=J7,width=1,height=7}{J6}{}
    \namedpnode[fill=white]{right of=J7,y-min=0,width=3,height=11}{J9}{}
    \namedpnode[fill=white]{right of=J9,y-min=0,width=2,height=14}{J2}{}
    \namedpnode[fill=white]{right of=J2,y-min=0,width=5,height=12}{J3}{}

    \namedpnode[fill=white]{x-min=11,y-min=0,width=4.5,height=7}{J1}{}
    \namedpnode[fill=white]{right of=J1,y-min=0,width=5,height=8}{J4}{}
    \namedpnode[fill=white]{right of=J4,y-min=0,width=1,height=8.5}{J5}{}
    \namedpnode[fill=white]{right of=J5,y-min=0,width=2,height=8}{J10}{}

    \namedpnode[fill=white]{x-min=11,y-min=10,width=7,height=4}{J8}{}
    \namedpnode[fill=white]{right of=J8,y-min=10,width=6,height=4.5}{J11}{}

    \pnode[draw=none,fill=blue,fill opacity=0.1,text opacity=0.75]{x-min=0,y-min=0,x-max=11,y-max=15}{\Large $\mathbf{S_0}$}
    \pnode[draw=none,fill=yellow,fill opacity=0.1,text opacity=0.75]{x-min=11,y-min=0,x-max=24,y-max=10}{\Large $\mathbf{S_1}$}
    \pnode[draw=none,fill=red,fill opacity=0.1,text opacity=0.75]{x-min=11,y-min=10,x-max=24,y-max=15}{\Large $\mathbf{S_2}$}

 	  \draw[|-|] (0,0) -- (0,10) node[pos=0,left] {\footnotesize $0$} node[pos=1,left] {\footnotesize $d$};
 	  \draw[-|] (0,10) -- (0,15) node[pos=1,left] {\footnotesize $\frac{3}{2}d$};
  \end{tikzpicture}
  \caption{An example of a schedule after applying the transformation rules}
  \label{fig:3-shelf-schedule}
\end{figure}
\begin{enumerate}
  \item If job $j \in S_1$ has processing time~$\ptm{j}{\procnum{j}{d}} \leq \frac{3}{4} d$
    and $\procnum{j}{d} > 1$, then allocate j to $\procnum{j}{d} - 1$ processors in $S_0$. \label{item:rule-1}
  \item If two jobs $j, j' \in S_1$ each have processing time at most~$\frac{3}{4} d$
    and $\procnum{j}{d} = \procnum{j'}{d} = 1$, then schedule $j$ and $j'$
    sequentially on one processor in $S_0$.
    If there is only one job~$j$ with $\ptm{j}{\procnum{j}{d}} \leq \frac{3}{4} d$ and $\procnum{j}{d} = 1$,
    a special case applies~\cite{mounie04}: if also a job~$j'$ in $S_1$ with $\ptm{j'}{\procnum{j'}{d}} > \frac{3}{4} d$
    and $\ptm{j}{\procnum{j}{d}} + \ptm{j'}{\procnum{j'}{d}} \leq \frac{3}{2} d$ exists,
    $j$ is scheduled on top of $j'$ in $S_0$.
    Conceptually, $j'$ is split into two parts. One is allocated to one processor in $S_0$,
    while the remaining part remains in $S_1$. \label{item:rule-2}
  \item Let $p_0$ and $p_1$ denote the number of processors required by $S_0$ and $S_1$, respectively.
    Let $q = m - (p_0 + p_1)$.
    If there is $j \in S_2$ with $\ptm{j}{q} \leq \frac{3}{2} d$,
    allocate $j$ with $p = \procnum{j}{\frac{3}{2} d} \leq q$ processors
    to $S_0$ if $\ptm{j}{p} > d$ and to $S_1$ otherwise. \label{item:rule-3}
\end{enumerate}

The resulting schedule is indeed feasible.
\begin{lemma}[\cite{mounie07}]
  \label{lemma:transform}
  When none of the transformation rules are applicable, the schedule uses at most~$m$ processors.
\end{lemma}

Since the transformation rules do not increase the number of processors allotted to any job,
the work does not increase.
In particular, if the two-shelf-schedule has total work at most $md - W_S(d)$,
then the resulting schedule also has total work at most $md - W_S(d)$.

Note that any job's allotment changes twice only if it is moved from~$S_2$ to~$S_1$ by rule~\labelcref{item:rule-3}
and then from~$S_1$ to~$S_0$ by rules~\labelcref{item:rule-1} or~\labelcref{item:rule-2}.
Therefore, the transformations can be exhaustively applied to a two-shelf-schedule in time~$\Landau{<=}{n \log n}$:
Scan the shelf~$S_1$ and classify each job into one of three categories
\begin{itemize}
  \item $\ptm{j}{\procnum{j}{d}} \leq \frac{3}{4} d$ and $\procnum{j}{d} > 1$,
  \item $\ptm{j}{\procnum{j}{d}} \leq \frac{3}{4} d$ and $\procnum{j}{d} = 1$, and
  \item $\ptm{j}{\procnum{j}{d}} > \frac{3}{4} d$.
\end{itemize}
Jobs of the first category can immediately be moved to~$S_0$ according to rule~\labelcref{item:rule-1}.
Jobs of the second category are stored and moved to~$S_0$ in pairs as described in rule~\labelcref{item:rule-2}.
This may leave one unpaired job an the end.
Jobs of the third category are stored in a min-heap with the processing time~$\ptm{j}{\procnum{j}{d}}$ as key.
This heap can be used to pair unpaired jobs from the second category
in accordance to the special case of rule~\labelcref{item:rule-2}.
Scanning shelf~$S_1$ in this fashion can be done in time~$\Landau{<=}{n \log n}$.
Afterwards, we process shelf~$S_2$ and check if rule~\labelcref{item:rule-3} applies.
If it does, and the job should be moved to~$S_1$, we immediately categorize it
and apply rules~\labelcref{item:rule-1} and~\labelcref{item:rule-2} to it.
This requires time~$\Landau{<=}{\log n}$ per job in~$S_2$.
We also need to compute some of the values $\procnum{j}{d}$, $\procnum{j}{\frac{d}{2}}$,
and $\procnum{j}{\frac{3}{2}d}$ for each job~$j \in J_B(d)$,
resulting in an overall running time of~$\Landau{<=}{n (\log n + \log m)}$.

\paragraph{Re-adding the Small Jobs}

We modify the schedule such that the free time on each processor is adjacent:
jobs in $S_0$ and $S_1$ start as early as possible, while the jobs in $S_2$ finish at time $\frac{3}{2} d$.
We now allocate the small jobs to the free time intervals with a next-fit approach\footnote{in the work of Mounié, Rapine, and Trystram~\cite{mounie07}, a different allocation rule is used,
but the proof is similar.}.
The small jobs are processed in an arbitrary order.
The current job~$j$ gets assigned to the next machine with load at most~$\frac{3}{2}d-\ptm{j}{1}$.

\begin{lemma}
  \label{lemma:small}
  Using a next-fit algorithm, all small jobs can be integrated into a three-shelf-schedule,
  provided its total work is at most $md - W_S(d)$. 
\end{lemma}

\begin{proof}
  Assume for the sake of contradiction that a small job cannot be added.
  Then all machines were discarded for not fitting some job,
  i.e.~their load is larger than~$\frac{3}{2}d - \ptm{j}{1}$ for a small job~$j$.
  Since $j$ is small, i.e.~$\ptm{j}{1} \leq \frac{d}{2}$,
  those machines load was larger than~$\frac{3}{2}d - \frac{1}{2}d = d$.
  Thus the total load of all machines is larger than~$md$, a contradiction.
\end{proof}

This can be implemented to run in linear time.
We first group adjacent processors that have the same amount of free time.
Since the amount of free time only changes when a different job is scheduled on a processor,
there can be at most $\Landau{<=}{n}$~many groups.
The groups are processed in arbitrary order.
If the free time of a group is at least~$\ptm{j}{1}$, where $j$ is the current job,
the group is split into two groups,
the first containing one processor, and the second containing the other processors.
The current job is then added to the processor in the first group, and its free is updated.
Otherwise the whole group is discarded.
This is possible in time linear in the number of small jobs and groups,
achieving an overall running time of~$\Landau{<=}{n}$.

\subsection{Knapsack with Compressible Items}
\label{sec:fully-polynomial}

The dominating part of the algorithm is the solution of the knapsack problem.
One might be tempted to use one of the known FPTASs for the knapsack problem
to find a solution with slightly suboptimal profit.
However, the profit of the knapsack problem can be much larger than the work of the schedule,
such that a small decrease of the profit
can increase the work of the schedule by a much larger factor.
Instead, we treat the processor counts approximately,
despite the fact that the the available number of processors imposes a hard constraint.
We will employ compression to compensate for an increased processor usage of the solution.

For a cleaner notation, we define the knapsack problem with compressible items:
an instance of this problem is a tuple~$(I, I^\mathrm{c}, C, \rho)$,
where $I$ is a set of items.
An item~$i \in I$ has size~$\s{i}$ and profit~$\pr{i}$.
The items $I^\mathrm{c} \subseteq I$ can be compressed with factor~$\rho$,
and $C$~denotes the capacity.
A feasible solution to this instance is a set $I' \subseteq I$ such that
\begin{equation}
  \sum_{i \in I' \cap I^\mathrm{c}} (1-\rho) \s{i}
    + \sum_{i \in I' \setminus I^\mathrm{c}} \s{i} \leq C \dpunct{.}
\end{equation}
We denote the maximum profit of the instance by $\KPCOPT{I, I^\mathrm{c}, C, \rho}$.

\subsubsection{Simple Application to the Scheduling Problem}
\label{sec:application}

Identifying the jobs as items is straightforward:
Let~$J^C = \setst{j \in J_B(d)}{\procnum{j}{d} \geq \frac{1}{\rho}}$ be the compressible jobs,
where the compression factor~$\rho$ will be defined later
depending on the desired accuracy~$\epsilon$.
Set the knapsack sizes as $\s{j} = \procnum{j}{d}$ and the profits as $\pr{j} = v_j(d)$ for $j \in J_B(d)$.
If we have a solution $J'$ to the knapsack problem $(J_B(d), J^C, m, \rho)$,
we can compress the jobs in $J' \cap J^C$ such that $J'$ fits on $m$ processors.
Their processing time increases by a factor of at most~$(1+4\rho)$,
so $J'$ is a feasible solution to $\KP{J_B(d)}{m, (1+4\rho)d}$.

To make up for the increased makespan,
we use the following corollary to Lemma~\ref{lemma:make-feasible}
with $d' = (1 + 4\rho)d$.
\begin{corollary}
  \label{thm:approx}
  Let $d' \geq d$, and~$J' \subseteq J_B(d)$
  be a feasible solution of the knapsack problem $\KP{J_B(d)}{m,d'}$
  with~$W(J', d) \leq md' - W_S(d)$.
  Then we can, in time~$\Landau{<=}{n \log n}$,
  find a schedule with makespan at most $\frac{3}{2}d'$.
\end{corollary}

\begin{proof}
  Let $J'' = J' \cap J_B(d')$.
  Then $J''$ is a feasible solution of $\KP{J_B(d')}{m,d'}$.
  Define $\Delta W = W_S(d') - W_S(d)$.
  Note that, when targeting makespan $\frac{3}{2}d'$,
  each job in $J' \setminus J''$ will be scheduled on one processor in shelf~1 and shelf~2,
  since these are exactly the jobs that have $\ptm{j}{1} \leq \frac{d'}{2}$.
  Therefore, they can be moved between shelves without affecting the work of the schedule,
  implying~$W(J'', d') = W(J', d')$.
  Using also the fact that the work functions are monotone, we obtain
  \begin{equation}
    \begin{aligned}
      W(J'', d') &= W(J', d') \\
      &= \sum_{j \in J'} \work{j}{\procnum{j}{d'}} + \sum_{\mathclap{j \in J_B(d') \setminus J'}} \work[2]{j}{\procnum[2]{j}{\textstyle\frac{d'}{2}}} \\
      &\leq \sum_{j \in J'} \work{j}{\procnum{j}{d}} + \sum_{\mathclap{j \in J_B(d') \setminus J'}} \work{j}{\procnum{j}{\textstyle\frac{d}{2}}} \\
      &\begin{multlined}
      = \sum_{j \in J'} \work{j}{\procnum{j}{d}} + \sum_{\mathclap{j \in J_B(d) \setminus J'}} \work{j}{\procnum{j}{\textstyle\frac{d}{2}}} \\
         - \sum_{\mathclap{j \in J_B(d) \setminus J_B(d')}} \work{j}{\procnum{j}{\textstyle\frac{d}{2}}} \end{multlined} \\
      &= W(J', d) - \Delta W \\
      &\leq md' - W_S(d) - \Delta W \\
      &= md' - W_S(d') \dpunct{.}
    \end{aligned}
  \end{equation}
  Applying Lemma~\ref{lemma:make-feasible} to $J''$ yields the desired schedule.
\end{proof}

To obtain a $\parens[2]{\frac{3}{2} + \epsilon}$-dual approximate algorithm,
we set $\rho = \frac{1}{6} \epsilon$.
Concerning the prerequisites of Corollary~\ref{thm:approx},
it is sufficient if the profit of $J'$ is at least $\KPCOPT{I, \emptyset, C, 0} = \KPOPT{J_B(d)}{m, d}$,
because then $W(J', d) \leq md - W_S(d)$, unless no schedule with makespan $d$ exists.

\begin{algorithm}
  \Input{$J, m, d, \epsilon$}
  $\rho \leftarrow \frac{1}{6}\epsilon$, $d' \leftarrow (1+4\rho)d$,
  $J^C \leftarrow \setst{j \in J_B(d)}{\procnum{j}{d} \geq \frac{1}{\rho}}$ \;
  \For{$j \in J_B(d)$}{
    Precompute $\procnum{j}{\frac{d}{2}}, \procnum{j}{d}$,
      $\procnum[2]{j}{\frac{d'}{2}}$, $\procnum{j}{d'}$, $\procnum{j}{\frac{3}{2}d'}$ \;
  }
  Find solution $J'$ to $(J_B(d), J^C, m, \rho)$ \label{line:knapsack}\;
  Apply Lemma~\ref{lemma:make-feasible} to $J' \cap J_B(d')$,
    obtain schedule for jobs $J$ with makespan $\frac{3}{2}d'$ \label{line:final-schedule}\;
  \lIf{schedule is infeasible}{
    \KwSty{reject} $d$}
  \lElse{
    \Return{the schedule}}
  
  \caption{Scheduling of monotone moldable jobs using knapsack with compressible items}
  \label{alg:scheduling-simple}
\end{algorithm}

Note that the algorithm schedules the jobs
with $\procnum{j}{d'}$ processors.
Compression was only used to show that
shelf~1 has at most~$m$ processors with these processor counts.

How one can find a solution to is discussed in the rest of this section.

\subsubsection{Separating Compressible and Incompressible Items}

Let $I = I_1 \cup I_2$ be a partition of the items.
Consider an optimal solution $I^* \subseteq I$ to the knapsack problem $(I, I^\mathrm{c}, C, \rho)$,
let $\alpha \geq \sum_{i \in I^* \cap I_1} \s{i}$ be the space available for items in $I_1$ in the solution,
and $\beta \geq \sum_{i \in I^* \cap I_2} \s{i}$ the space available for the other items.
Then we can solve the knapsack problems for $I_1$ and $I_2$ separately.

\begin{lemma}
  \label{lemma:knapsack-separate}
  $\KPCOPT{I, I^\mathrm{c}, C, \rho} \leq \KPCOPT{I_1, I^\mathrm{c} \cap I_1, \alpha, \rho}
  + \KPCOPT{I_2, I^\mathrm{c} \cap I_2, \beta, \rho}$,
  and equality holds if $\alpha + \beta = C$.
\end{lemma}

\begin{proof}
  Let~$I_1^* \subseteq I_1$ and~$I_2^* \subseteq I_2$ be optimal solutions
  of~$(I_1, I^c \cap I_1, \alpha, \rho)$ and $(I_2, I^c \cap I_2, \beta, \rho)$, respectively.
  Assume for the sake of contradiction that
  \begin{equation}
    \KPCOPT{I, I^c, C, \rho} > \KPCOPT{I_1, I^c \cap I_1, \alpha, \rho} + \KPCOPT{I_2, I^c \cap I_2, \beta, \rho} \dpunct{.}
  \end{equation}
  Recall that $I^*$~is our optimal solution to~$(I, I^c, C, \rho)$.
  Then
  \begin{equation}
    \begin{aligned}
      \sum_{i \in I^* \cap I_1} & \pr{i} + \sum_{i \in I^* \cap I_2} \pr{i} \\
      &= \sum_{i \in I^*} v_j(d) = \KPCOPT{I, I^c, C, \rho} \\
      &> \KPCOPT{I_1, I^c \cap I_1, \alpha, \rho} + \KPCOPT{I_2, I^c \cap I_2, \beta, \rho} \dpunct{,}
    \end{aligned}
  \end{equation}
  therefore we either have $\sum_{i \in I^* \cap I_1} \pr{i} > \KPCOPT{I_1, I^c \cap I_1, \alpha, \rho}$
  and $\sum_{i \in I^* \cap I_1} \s{i} \leq \alpha$,
  or $\sum_{i \in I^* \cap I_2} \pr{i} > \KPCOPT{I_2, I^c \cap I_2, \beta, \rho}$
  and $\sum_{i \in I^* \cap I_2} \s{i} \leq \beta$,
  contradicting the optimality of $I_1^*$ and $I_2^*$.
  
  If $\alpha + \beta = C$,
  then the set~$I_1^* \cup I_1^*$ is a feasible solution of~$(I, I^c, C, \rho)$,
  hence
  \begin{equation}
    \KPCOPT{I, I^c, C, \rho} \geq \KPCOPT{I_1, I^c \cap I_1, \alpha, \rho} + \KPCOPT{I_2, I^c \cap I_2, \beta, \rho}
  \end{equation}
  and equality must hold.
\end{proof}

In our case, we partition $I = I^c \cup (I \setminus I^c)$,
i.e. $\alpha$ is the space available for compressible items.
To utilize this lemma, we need to know values for~$\alpha$ and~$\beta$.
Let $\beta_{\max}$ be an upper bound on the space used by incompressible jobs.
We can enumerate all $\beta_{\max} + 1$ possible values of~$\beta \leq \beta_{\max}$,
set $\alpha = C - \beta$,
and pick the best obtained solution.
A standard dynamic programming algorithm can be used to solve
the instances~$(I \setminus I^\mathrm{c}, \emptyset, \beta, 0)$
in time~$\Landau{<=}{n \beta_{\max}}$ each.

While the problem for the compressible items could be solved the same way,
we are looking for an algorithm that is polynomial in $\log C$.
We propose several techniques, exploiting the compressibility of the items.
First, we treat the size of compressible items approximately.
Second, we further bound the number of knapsack problems to solve
by~$\Landau{<=}{\frac{\log C}{\rho}}$
by using an approximate, possibly larger value for $\alpha$.
We also solve all knapsack problems in one pass.

\subsubsection{Solving the Knapsack Problem for Compressible Items}
\label{sec:dp-knapsack}

One way to implement the dynamic program for solving the knapsack problem exactly is given by Lawler~\cite{lawler79}:
assume there are $n_C \leq n$ compressible items~$\irange{i}{1}{n_C}$, and the capacity is $\alpha$.
A list~$L$ of pairs $(p, s)$ is initialized with the single pair $(0, 0)$.
In the $k$-th iteration, for each pair~$(p, s)$ in~$L$,
a new pair~$(p + \pr{i_k}, s + \s{i_k})$ is added to~$L$,
unless $s + \s{i_k} > \alpha$.
Thus, at the end of the $k$-th iteration, a pair~$(p, s)$
indicates that $p$ is the highest profit that can be obtained with the items~$\irange{i}{1}{k}$
and total size at most~$s$.
A pair $(p, s)$ is said to dominate~$(p', s')$ if $p \geq p'$ and $s \leq s'$.
After each modification, we remove dominated pairs from the list.
The optimum value then is~$\max \setst{p}{(p,w) \in L}$.
By storing additional backtracking information,
an optimal solution can be found in time~$\Landau{<=}{n_C \alpha} = \Landau{<=}{n_C m}$.

Let $\bar{n}$ be an upper bound to the number of compressible items in any solution.
In our scheduling setting,
such a bound is imposed by the fact that wide jobs are compressible.
A common approach is to round the sizes $\s{i}$ and the capacity~$\alpha$
down to the next multiple of $U = \frac{\rho}{(1-\rho) \bar{n}} \alpha$,
such that all sizes in pairs stored in~$L$ are multiples of~$U$.
This can be equivalently achieved as follows:
cover the range $0, \dotsc, \alpha$ with disjoint intervals of length~$U$,
i.e.~$[0, \alpha] \subseteq [0, U) \cup [U, 2U) \cup [2U, 3U) \cup \dots$.
This requires $\ceil{\frac{\alpha}{U}} = \Landau{<=}{\bar{n}}$ intervals.
On creation, the pairs $(p, s)$ are normalized
such that $s \in [\ell U, (\ell+1)U)$ implies $s = \ell U$,
i.e.~the value of $s$ is reduced by at most~$U$ such that it is a multiple of~$U$.
The profit of an optimal solution to the rounded instance can be calculated as~$\max \setst{p}{(p, s) \in L}$.
Since the actual width of the items does not decrease,
the solution may be up to $\bar{n} U$ units larger than the capacity,
but is small enough when we take into account that all items are compressible:
\begin{equation}
  (1 - \rho) \times \parens{\alpha + \bar{n} U}
    = (1 - \rho) \times \parens[a]{1 + \frac{\rho}{(1-\rho)}} \alpha = \alpha \dpunct{.}
  \label{eq:knapsack-error}
\end{equation}
The required running time of the algorithm is~$\Landau{<=}{n_C \bar{n}}$.

\subsubsection{Solving the Knapsack Problems in One Pass}
\label{sec:one-pass}

We first demonstrate how to solve the knapsack problems $(I \setminus I^\mathrm{c}, \emptyset, \beta, 0)$
for each~$\beta$ in some set $B$ in one pass.
For this we modify the dynamic programming approach by Lawler outlined in \cref{sec:dp-knapsack}
to solve the knapsack problem for several capacities.
Similar to the original algorithm, we build the list~$L$ of pairs~$(p, s)$,
but up to the largest capacity~$\max B$.
This requires time~$\Landau{<=}{n_I \max B}$,
where $n_I = n - n_C$ is the number of incompressible items.
A pair~$(p, s)$ in the list means that~$p$ is the best obtainable profit with capacity~$s$. 
For each $\beta \in B$ we now find the largest~$s \leq \beta$
such that a pair~$(p, s)$ is in the list.
Then $p$~is the optimal profit for capacity~$\beta$.
Thus all knapsack problems can be solved in time~$\Landau{<=}{n_I \max B}$.

For compressible items,
we additionally use the normalization technique from \cref{sec:dp-knapsack}.
Let $\alpha_{\min} > 0$ be a lower bound any non-zero~$\alpha$,
e.g.~the minimum size of a compressible item.
When solving a knapsack problem for several capacities, we cannot normalize all sizes in the same way.
Instead, we use adaptive normalization.
The idea is illustrated in \cref{fig:adaptive-rounding}.
To bound the running time, we also need to impose a requirement on the set of capacities.

\begin{lemma}
  \label{thm:adaptive-rounding}
  Let $A = \set{\alpha_1 < \cdots < \alpha_k}$~be a set of $k$~capacities such that
  for each $i \in \setrange{1}{k}$ we have
  \begin{equation}
    \label{eq:gap-alpha}
    \alpha_i - \alpha_{i-1} \leq \rho \alpha_i \dpunct{,}
  \end{equation}
  where $\alpha_0 = \alpha_{\min}$.
  Then all knapsack problems~$(I^c, I^c, \alpha, \rho)$ with~$\alpha \in A$ can be solved 
  with profit at least $\KPCOPT{I^c, \emptyset, \alpha, 0}$
  in time~$\Landau{<=}{n_C \bar{n} \card{A}}$.
\end{lemma}

\begin{proof}
  Partition the interval~$[\alpha_{\min}, \alpha_k]$ into intervals~$I^{(1)}, \dotsc, I^{(k)}$.
  For this, define~$I^{(i)} = [\alpha_{i-1}, \alpha_i)$ for each $i \in \setrange{1}{k}$.
  Further partition each interval~$I^{(i)}$
  into subintervals, similar as in \cref{sec:dp-knapsack}:
  for $i \in \setrange{1}{k}$ set~$U_i = \frac{\rho}{(1-\rho) \bar{n}} \alpha_i$,
  and let~$I^{(i)}_{\ell} = [\ell U_i, (\ell + 1) U_i) \cap I^{(i)}$
  for~$\ell \in \setrange{\ell^{(i)}_{\min}}{\ell^{(i)}_{\max}}$.
  We choose~$\ell^{(i)}_{\min} = \floor{\frac{\alpha_{i-1}}{U_i}}$
  and~$\ell^{(i)}_{\max} = \floor{\frac{\alpha_i}{U_i}}$.
  Then $I^{(i)}_{\ell^{(i)}_{\min}}$ and~$I^{(i)}_{\ell^{(i)}_{\max}}$
  are the subintervals that start respectively end at $\alpha_{i-1}$ and $\alpha_i$.
  See \cref{fig:adaptive-rounding} for a visualization of the interval structure.
  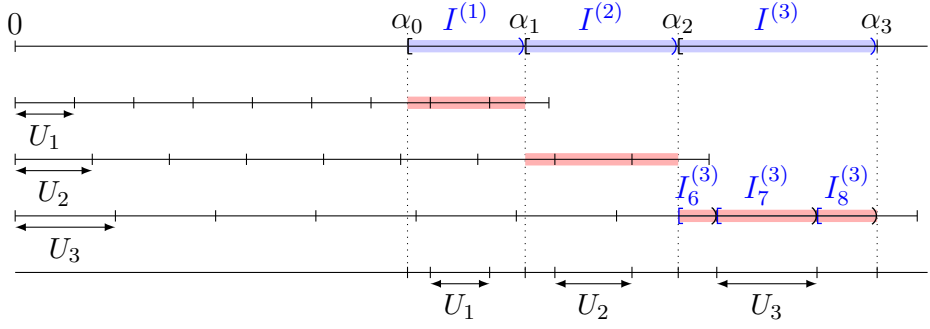
\begin{figure*}
    \centering
    \begin{tikzpicture}[x=1.2cm,y=0.75cm]
      \newcommand{\mymark}[3][]{
        \draw ($(#2,#3+0.1)$) -- ($(#2,#3-0.1)$);
        \notblank{#1}{
          \node at ($(#2,#3+0.4)$) {#1};
        }{}
      }

      \renewcommand{\k}{3} 

      % Draw line for alpha's
      \draw (0,0) -- (10,0);
      \mymark[$0$]{0}{0}

      % Draw line for intervals
      \pgfmathsetmacro{\gy}{-\k-1}
      \draw (0,\gy) -- (10,\gy);
      
      % Draw intermediate intervals for each alpha
      \foreach \i in {0,...,\k} {
        \pgfmathsetmacro{\a}{4.3*1.3^\i}
        \pgfmathsetmacro{\U}{\a/8.6}
        \pgfmathsetmacro{\y}{-\i}
        
        % Mark alpha and dotted line
        \mymark[$\alpha_\i$]{\a}{0}
        \draw[dotted] (\a,0.23) -- (\a,\gy);
        
        % Highlight interval I^(i)
        \ifnumequal{\i}{0}{
        }{
          \pgfmathsetmacro{\leftbound}{\a/1.3}
          \pgfmathsetmacro{\x}{\a-0.015}
          \draw[-{Parenthesis[blue,line width=0.5pt,width=1.5ex,length=0.7ex]},draw=blue!20,line width=0.9ex] (\leftbound,0) -- node[midway,above,blue] {$I^{(\i)}$} (\x,0);
          \draw[-{Bracket[blue,line width=0.5pt,width=1.2ex,length=0.3ex]}] (\a,0) -- (\leftbound,0);
        }

        \ifnumequal{\i}{0}{
          \mymark{\a}{\gy}
        }{
          \pgfmathsetmacro{\lmin}{floor(\a/1.3/\U)}
          \pgfmathsetmacro{\lmax}{floor(\a/\U)}

          % Draw intermediate line
          \draw[red!30,line width=0.9ex] (\leftbound,\y) -- (\a,\y);
          \pgfmathsetmacro{\x}{(\lmax+1)*\U}
          \draw (0,\y) -- (\x,\y);

          % Mark subiterval bounds
          \pgfmathsetmacro{\num}{\lmax+1}
          \foreach \j in {0,...,\num} {
            \pgfmathsetmacro{\x}{\j*\U}
            \mymark{\x}{\y}
          }
          
          % Highlight subintervals of I^(k)
          \ifnumequal{\i}{\k}{
            \foreach \j in {\lmin,...,\lmax} {
              \pgfmathsetmacro{\j}{int(\j)}
              \pgfmathsetmacro{\leftbound}{max(\j*\U, \a/1.3)}
              \pgfmathsetmacro{\rightbound}{min((\j+1)*\U, \a)}
              \draw[-{Bracket[blue,line width=0.5pt,width=1.2ex,length=0.3ex]},draw=blue!10,line width=0.0ex] (\rightbound,\y) -- node[midway,above,blue] {$I^{(\i)}_{\j}$} (\leftbound,\y);
              \draw[-{Parenthesis[blue,line width=0.5pt,width=1.5ex,length=0.7ex]}] (\leftbound,\y) -- (\rightbound,\y);
            }
          }{}          
          
          % Mark subinterval length on intermediate line
          \pgfmathsetmacro{\y}{\y - 0.2}
          \draw[<->,>=latex] (0,\y) -- (\U,\y)
            node[midway,below]{\small $U_\i$};

          % Mark interval bounds on global line
          \pgfmathsetmacro{\num}{\lmin+1}
          \foreach \j in {\num,...,\lmax} {
            \pgfmathsetmacro{\x}{\j*\U}
            \mymark{\x}{\gy}
          }
          \mymark{\a}{\gy}

          % Mark subinterval length on global line
          \pgfmathsetmacro{\x}{\num*\U}
          \pgfmathsetmacro{\xx}{\x + \U}
          \pgfmathsetmacro{\y}{\gy - 0.2}
          \draw[<->,>=latex] (\x,\y) -- (\xx,\y)
            node[midway,below]{\small $U_\i$};
        }
      }
    \end{tikzpicture}
    \caption{Example of the interval structure used for the normalization.
    Here, $\ell^{(i)}_{\min} = 6$ and $\ell^{(i)}_{\max} = 8$.
    The top line shows the partitioning into intervals.
    On the intermediate lines the partitioning from \cref{sec:dp-knapsack}
    for a maximum capacity of $\alpha_i$ is shown for each $i$.
    The highlighted parts are merged for the final partition,
    displayed on the bottom line.}
    \label{fig:adaptive-rounding}
  \end{figure*}
  We again use a dynamic program to calculate the list~$L$ of pairs~$(p, s)$,
  and normalize the widths such that $s \in I^{(i)}_{\ell}$
  implies~$s = \min I^{(i)}_{\ell} = \max(\ell U_i, \alpha_{i-1})$.
  Then we can, for each $i \in \setrange{1}{k}$, compute the profit of a solution
  to~$(I^c, I^c, \alpha_i, \rho)$ as~$\max \setst{p}{\text{$(p,s) \in L$ and $s \leq \alpha_i$}}$:
  Since $U_1 \leq \dots \leq U_k$, the subintervals in $[0, \alpha_i]$ have width at most~$U_i$.
  Therefore, the width of a solution for a capacity in this range
  is underestimated by at most~$\bar{n} U_i$,
  which the compression compensates for, see also \cref{eq:knapsack-error}.
  
  To determine the running time, we need to count the number of subintervals.
  For $i \in \setrange{1}{k}$,
  the interval~$I^{(i)}$ has $\ell^{(i)}_{\max} - \ell^{(i)}_{\min} + 1$ many subintervals.
  Therefore,
  \begin{equation}
    \begin{aligned}
      \ell^{(i)}_{\max} - \ell^{(i)}_{\min}
      &= \floor[a]{\frac{\alpha_i}{U_i}} - \floor[a]{\frac{\alpha_{i-1}}{U_i}} \\
       &\leq \frac{\alpha_i}{U_i} - \parens[a]{\frac{\alpha_{i-1}}{U_i} - 1} \\
       &\leq \frac{(\alpha_i - \alpha_{i-1}) (1 - \rho)}{\rho \alpha_i} \bar{n} + 1 \\
       &\leq \frac{\rho \alpha_i (1 - \rho)}{\rho \alpha_i} \bar{n} + 1
         = (1 - \rho) \bar{n} + 1 \\
       &= \Landau{<=}{\bar{n}} \dpunct{.}
    \end{aligned}
  \end{equation}
  In total, we have $\Landau{<=}{\bar{n} \card{A}}$~many intervals.
  The dynamic program therefore has running time~$\Landau{<=}{n_C \bar{n} \card{A}}$.
\end{proof}

\subsubsection{Reducing the Number of Knapsack Problems}
\label{sec:reducing-kp-problems}

We propose to approximate the space~$\alpha$ available for compressible items
with a value~$\tilde{\alpha} \geq \alpha$,
and use half of the compressibility for this.
Assume that $\alpha > 0$,
the case $\alpha = 0$ will be handled separately.
The uncompressed items in a solution of~$(I^c, I^c, \tilde{\alpha}, \rho)$
have size at most $\frac{\tilde{\alpha}}{1-\rho}$.
Compressed with factor $\rho' = 2\rho - \rho^2$,
they have a size of at most~$\frac{1-2\rho + \rho^2}{1-\rho} \tilde{\alpha} = (1-\rho) \tilde{\alpha}$.
We therefore require $\tilde{\alpha}$ to satisfy
\begin{equation}
  \label{eq:rounded-alpha}
  \alpha \leq \tilde{\alpha} \leq \frac{1}{1-\rho} \alpha \dpunct{,}
\end{equation}
and will construct a set of such values by using a geometric progression.

\begin{definition}
  For any positive reals $L$, $U$, and $x > 1$
  we define
  \begin{equation}
    \textstyle \mathop{\mathrm{geom}}(L, U, x) = \setst{Lx^i}{i \in \range{0}{\ceil{\log_{x}\frac{U}{L}}}} \dpunct{.}
  \end{equation}
\end{definition}

\begin{lemma}
  \label{lemma:geometric-rounding}
  For any $L \leq U$, and $1 < x < 2$ we have $\card{\mathop{\mathrm{geom}}(L, U, x)} = \Landau{<=}{\frac{1}{x-1} \log \frac{U}{L}}$.
\end{lemma}

\begin{proof}
  Define $y = x - 1 > 0$ and
  \begin{equation}
    i_{\max} = \floor[a]{2 \frac{\log \frac{U}{L}}{y}} + 1 \dpunct{.}
  \end{equation}
  We claim that $L x^{i_{\max}} > U$, or equivalently, $(1+y)^{i_{\max}} > \frac{U}{L}$.
  
  Since $x < 2$ implies $y < 1$, we have
  \begin{equation}
    i_{\max} > 2 \frac{\log \frac{U}{L}}{y} \geq \frac{\log \frac{U}{L}}{\frac{y}{1+y}} \dpunct{.}
  \end{equation}
  Now consider that $\frac{y}{1 + y} \leq \log (1+y)$ for $y \geq 0$:
  It holds for $y = 0$,
  and we have
  \begin{equation}
    \od{}{y} \frac{y}{1 + y} = \frac{1}{(1+y)^2} \leq \frac{1}{1 + y} = \od{}{y} \log (1+y)
  \end{equation}    
  for $y \geq 0$.
  Therefore
  \begin{equation}
    i_{\max} > \frac{\log \frac{U}{L}}{\frac{y}{1+y}}
    \geq \frac{\log \frac{U}{L}}{\log (1 + y)}
    = \log_{1+y} \frac{U}{L} \dpunct{,}
  \end{equation}
  which proves the claim.
  Hence $\mathop{\mathrm{geom}}(L, U, x)$ has cardinality at most $i_{\max} = \Landau{<=}{\frac{1}{y} \log \frac{U}{L}}$.
\end{proof}

Recall that $\alpha_{\min}$ is a lower bound on any $\alpha > 0$.
Let $A = \mathop{\mathrm{geom}}(\alpha_{\min} \frac{1}{1-\rho}, C, \frac{1}{1-\rho})$,
then $A$ contains an~$\tilde{\alpha}$ with property~\eqref{eq:rounded-alpha}
for each $\alpha \in \setrange{\alpha_{\min}}{C}$.
Furthermore, Lemma~\ref{lemma:geometric-rounding} yields $\card{A} = \Landau{<=}{\frac{1}{\rho} \log \frac{C}{\alpha_{\min}}}$.

Putting everything together, we obtain \cref{alg:knapsack-compressible},
a fully polynomial algorithm for the knapsack problem with compressible items.

\begin{algorithm}
  \Input{$I, I^c, C, \rho, \alpha_{\min}, \beta_{\max}, \bar{n}$}
  $\alpha_{\min} \leftarrow \max(\alpha_{\min}, C - \beta_{\max})$ \label{line:alpha_min}\;
  $A \leftarrow \mathop{\mathrm{geom}}(\alpha_{\min} \frac{1}{1-\rho}, C, \frac{1}{1-\rho})$,
    $A_0 \leftarrow A \cup \set{0}$ \;
  \lFor{$\tilde{\alpha} \in A$}{
    $\mathop{\beta}(\tilde{\alpha}) \leftarrow C - (1 - \rho) \tilde{\alpha}$
  }
  $\mathop{\beta}(0) \leftarrow \beta_{\max}$,
    $B \leftarrow \setst{\mathop{\beta}(\tilde{\alpha})}{\tilde{\alpha} \in A_0}$ \;
  Solve $(I \setminus I^c, \emptyset, \beta, 0)$ for $\beta \in B$ \label{line:solve-incompressible}\;
  Solve $(I^c, I^c, \tilde{\alpha}, \rho)$ for $\tilde{\alpha} \in A_0$,
    solution for $\tilde{\alpha} = 0$ is $\emptyset$ \label{line:solve-compressible}\;
  \For{$\tilde{\alpha} \in A_0$}{
    Combine solutions for $(I \setminus I^c, \emptyset, \mathop{\beta}(\tilde{\alpha}), 0)$
      and $(I^c, I^c, \tilde{\alpha}, \rho)$ \label{line:combine-solutions}\;
  }
  \Return{the best combined solution}
  
  \caption{Solve knapsack with compressible items}
  \label{alg:knapsack-compressible}
\end{algorithm}

\begin{theorem}
  \label{thm:knapsack-compressible}
  \Cref{alg:knapsack-compressible} finds a solution to the instance $(I, I^c, C, \rho')$
  with profit at least $\KPCOPT{I, \emptyset, C, 0}$
  in time $\Landau{<=}{n_I \beta_{\max} + n_C \bar{n} \frac{1}{\rho'} \log \frac{C}{\alpha_{\min}}}$.
\end{theorem}

\begin{proof}
  First note that the choice of $B$ is reasonable,
  since $\tilde{\alpha} \leq \frac{1}{1-\rho}C$ for $\tilde{\alpha} \in A$,
  implying $\mathop{\beta}(\tilde{\alpha}) = C - (1 - \rho) \tilde{\alpha} \geq 0$.
  We now show that each combined solution is a feasible solution for $(I, I^c, C, \rho')$.
  Let $I' \subseteq I$ be a solution created in \cref{line:combine-solutions} for some fixed $\tilde{\alpha}$.
  Then $\s{I' \cap I^c} \leq \frac{1}{1-\rho} \tilde{\alpha}$
  and $\s{I' \setminus I^c} \leq \mathop{\beta}(\tilde{\alpha}) = C - (1 - \rho) \tilde{\alpha}$.
  Compressing items from $I^c$ with factor $\rho'$ leads to a total size of at most
  \begin{equation}
      (1-2\rho+\rho^2) \frac{1}{1-\rho} \tilde{\alpha} + C - (1-\rho) \tilde{\alpha}
      = C \dpunct{.}
  \end{equation}
  
  For the profit,
  consider an optimal solution $I^*$ of $(I, \emptyset, C, 0)$.
  Now define $\alpha = \sum_{i \in I^* \cap I^c} \s{i}$
  and $\beta = \sum_{i \in I^* \setminus I^c} \s{i}$.
  We claim that there is an $\tilde{\alpha} \in A_0$ such that
  $\tilde{\alpha} \geq \alpha$ and $\mathop{\beta}(\tilde{\alpha}) \geq \beta$.
  Recall that we enforced $\alpha_{\min} \geq C - \beta_{\max}$ in \cref{line:alpha_min},
  but this is not a restriction,
  since there always must be $C - \beta_{\max}$ space available for compressible jobs.
  If $\alpha = 0$, $\tilde{\alpha} = 0$ clearly satisfies the claim.
  If $0 < \alpha < \alpha_{\min}$,
  then $\tilde{\alpha} = \min A = \frac{1}{1-\rho} \alpha_{\min}$ is larger than~$\alpha$,
  and since there is $\alpha_{\min}$ space available for compressible jobs,
  $\mathop{\beta}(\tilde{\alpha}) = C - \alpha_{\min}$ space must suffice for the incompressible jobs.
  Otherwise we have $\alpha \geq \alpha_{\min}$,
  so there is one $\tilde{\alpha} \in A$ which satisfies \cref{eq:rounded-alpha}.
  Furthermore, $\mathop{\beta}(\tilde{\alpha}) = C - (1-\rho) \tilde{\alpha} \geq C - \alpha \geq \beta$. 
  According to \cref{thm:adaptive-rounding},
  the profit of the found solution is at least
  $\KPCOPT{I^c, \emptyset, \tilde{\alpha}, 0} + \KPCOPT{I \setminus I^c, \emptyset, \mathop{\beta}(\tilde{\alpha}), 0}$.
  Lemma~\ref{lemma:knapsack-separate} now proves that the profit is at least $\KPCOPT{I, \emptyset, C, 0}$.
  
  Regarding the running time,
  the definition of~$A$ clearly satisfies \cref{eq:gap-alpha}.
  Therefore, we can apply the methods described in \cref{sec:one-pass}
  for \cref{line:solve-incompressible,line:solve-compressible}.
  Because we ensured $\alpha_{\min} \geq C - \beta_{\max}$
  and $\tilde{\alpha} \geq \frac{1}{1-\rho} \alpha_{\min}$ for each $\tilde{\alpha} \in A$,
  $\mathop{\beta}(\tilde{\alpha}) = C - (1-\rho) \tilde{\alpha} \leq C - \alpha_{\min} \leq \beta_{\max}$.
  Therefore, $\max B \leq \beta_{\max}$
  and \cref{line:solve-incompressible} requires time~$\Landau{<=}{n \beta_{\max}}$.
  \Cref{line:solve-compressible} requires time~$\Landau{<=}{n \bar{n} \frac{1}{\rho} \log \frac{C}{\alpha_{\min}}}$.
  These steps clearly dominate the running time of the algorithm,
  and $\rho = \Landau{=}{\rho'}$.
\end{proof}

Now that we have an efficient means to solve the knapsack problem with compressible items,
we can create a simple fully polynomial algorithm based on \cref{alg:scheduling-simple}.
To use \cref{alg:knapsack-compressible},
we need to provide good bounds on~$\alpha_{\min}$, $\beta_{\max}$ and~$\bar{n}$.
Therefore, we only use \cref{alg:scheduling-simple} if $m < 16 n$.
Otherwise, if $m \geq 16 n$, we can use our FPTAS with $\epsilon = \frac{1}{2}$
to obtain a schedule with makespan~$\frac{3}{2} d$.
Then choose $\alpha_{\min} = \frac{1}{\rho}$, $\beta_{\max} = m = \Landau{<=}{n}$,
and $\bar{n} = m \rho = \Landau{<=}{\epsilon n}$.
Solving the knapsack problem thus requires $\Landau{<=}{n^2 \log \epsilon m}$~operations.
Using the PTAS or precomputing the required values of $\gamma_j$ requires $\Landau{<=}{n \log m}$~time.
Finding the schedule in \cref{line:final-schedule} requires $\Landau{<=}{n \log n}$~time
using Lemma~\ref{lemma:make-feasible}.
The dual algorithm has running time~$\Landau{<=}{n (\log m + n \log \epsilon m)}$ in total.

\subsection{The Improved Algorithm}
\label{sec:bounded-kp}

In the bounded knapsack problem, the input defines $k$ item types with sizes and profits,
and an item count $c_t$ for each item type~$t \in \natupto{k}$.
The number of items can be much larger than the number of item types.
An instance of the bounded knapsack instance with $k$ item types and capacity~$C$
can be transformed into an instance of the regular knapsack problem
with $\Landau{<=}{\log m}$ items per type~\cite{kellerer04}.
Each of these items serves as a container for an integer number
of items of the same type.

To further speed up the solution of the knapsack problem $(J_B(d), J^C, m, \rho)$,
we transform it into a bounded knapsack problem.
Beforehand, we reduce the number of item types by rounding the jobs.
But first, we define the threshold~$b$ for compressible jobs.
We also introduce an accuracy parameter~$\delta$, which we later choose depending on~$\epsilon$.

\begin{lemma}
  \label{lemma:compress-twice}
  Let $\delta \in (0,1]$, $\rho = \frac{1}{4}(\sqrt{1+\delta}-1)$, and $b = \frac{1}{2\rho - \rho^2}$.
  Any job that uses at least~$b$ processors can be compressed with factor~$2\rho - \rho^2$,
  decreasing its processor count by a factor $(1-\rho)^2$
  and increasing its processing time by a factor of less than~$1 + \delta$.
  Furthermore we have $\rho = \Landau{=}{\delta}$ and $b = \Landau{=}{\frac{1}{\delta}}$.
\end{lemma}

\begin{proof}
  Since $\delta \leq \frac{5}{4}$ we have $2\rho \leq \frac{1}{4}$, so $2\rho - \rho^2$ is a valid compression factor.
  Compression reduces the processor count by a factor~$(1 - 2\rho + \rho^2 = (1-\rho)^2$.
  The processing time increases by a factor of
  \begin{equation}
    1 + 4(2\rho - \rho^2) = 1 + 8\rho - 4\rho^2 < (1 + 4\rho)^2 = 1 + \delta \dpunct{.}
  \end{equation}    
  
  The identity $(1+4\rho)^2 = 1 + \delta$ implies
  $\rho = \frac{\delta}{8(1+2\rho)} \geq \frac{\delta}{12}$,
  since $\rho \leq \frac{1}{4}$.  
  Also, since $\sqrt{1+\delta} \leq 1+\delta$ we have $\rho \leq \frac{1}{4}\delta$,
  hence $\rho = \Landau{=}{\delta}$.

  Finally, $b = \frac{1}{2\rho - \rho^2} = \Landau{=}{\frac{1}{\rho}} = \Landau{=}{\frac{1}{\delta}}$,
  because $0 < \rho^2 \leq \frac{1}{4}\rho$,
  so $\frac{7}{4} \rho \leq 2\rho - \rho^2 < 2 \rho$.
\end{proof}

To define rounded sizes and profits for the items (jobs in $J_B(d)$),
we introduce a notation to round values geometrically;
$\mathop{\mathrm{\check{gr}}}(a, L, U, x) = \max \setst{a' \in \mathop{\mathrm{geom}}(L, U, x)}{a' \leq a}$
for rounding down, and $\mathop{\mathrm{\hat{gr}}}$ analogously for rounding up.
For $s \in \set{\frac{d}{2}, d}$, round the processor counts down
\begin{equation}
  \check{\gamma}_j(s) =
    \begin{cases*}
      \procnum{j}{s} & if $\procnum{j}{s} \leq b$ \\
      \mathop{\mathrm{\check{gr}}}(\procnum{j}{s}, b, m, 1 + \rho) & otherwise
    \end{cases*}
\end{equation}
and let $\s{j} = \check{\gamma}_j(d)$.
If $\check{\gamma}_j(\frac{d}{2}) < b$, round the original profit~$v_j(d)$ to
\begin{equation}
  \pr{j} =
    \begin{cases*}
      0 & if $v_j(d) < \frac{\delta}{2}d$ \\
      \mathop{\mathrm{\hat{gr}}}(v_j(d), \frac{\delta}{2}d, \frac{b}{2}d, 1 + \frac{\delta}{b}) & otherwise.
    \end{cases*}
\end{equation}
Otherwise, when $\check{\gamma}_j(\frac{d}{2}) \geq b$,
we consider rounded processing times for $s \in \set{\frac{d}{2}, d}$,
$\check{t}_j(s) = \mathop{\mathrm{\check{gr}}}(\ptm{j}{\procnum{j}{s}}, \frac{s}{2}, s, 1 + 4\rho)$
and set the profit as the saved work according to the rounded values,
i.e.~$\pr{j} = \check{t}_j(\frac{d}{2}) \check{\gamma}_j(\frac{d}{2}) - \check{t}_j(d) \check{\gamma}_j(d)$.

\subsubsection{Bounding the Number of Item Types}

In the first step, we round the processor counts $\procnum{j}{d}$ and $\procnum{j}{\frac{d}{2}}$.
Counts larger than $b$ are rounded down to the next value from $\mathop{\mathrm{geom}}(b, m, 1 + \rho)$.
According to Lemma~\ref{lemma:geometric-rounding},
this leaves us with $\Landau{<=}{\frac{\log (\delta m)}{\delta}}$~many different counts.

First, consider the types of compressible jobs,
i.e.~jobs that have processor count~$\procnum{j}{d} \geq b$.
\begin{lemma}
  \label{lemma:heights-bound}
  If we round the processing times $\ptm{j}{\procnum{j}{s}}$, where $s = d$ or $s = \frac{d}{2}$,
  down to the next value in $\mathop{\mathrm{geom}}(\frac{s}{2}, s, 1 + 4\rho)$  
  then there are at most $\Landau{<=}{\frac{1}{\delta}}$ many different rounded processing times (heights).
\end{lemma}

\begin{proof}
  Consider that the height of each job is more than half of the shelf height~$s$:
  Assume job $j$ has $\ptm{j}{\procnum{j}{s}} \leq \frac{s}{2}$ for the sake of contradiction.
  Then $\procnum{j}{s} > 1$, otherwise it would be a small job, which we removed.
  We claim that reducing the processor count by $1$ at most doubles the processing time,
  because we at most halve the processor count.
  More formally, we have
  $\procnum{j}{s} \leq 2 (\procnum{j}{s} - 1)$, because $\procnum{j}{s} \geq 2$, and
  \begin{equation}
    \begin{multlined}
      (\procnum{j}{s} - 1) \ptm{j}{\procnum{j}{s} - 1} = \work{j}{\procnum{j}{s} - 1} \\
      \leq \work{j}{\procnum{j}{s}} = \procnum{j}{s} \ptm{j}{\procnum{j}{s}}
    \end{multlined}
  \end{equation}    
  due to monotony.
  Therefore
  \begin{equation}
    \ptm{j}{\procnum{j}{s} - 1} \leq \frac{\procnum{j}{s}}{\procnum{j}{s} - 1} \ptm{j}{\procnum{j}{s}}
    \leq 2 \ptm{j}{\procnum{j}{s}} \dpunct{.}
  \end{equation}
  It follows that~$\ptm{j}{\procnum{j}{s}-1} \leq s$, a contradiction to the definition of~$\procnum{j}{s}$.
    
  The distance between two rounded processing times then must be
  at least $\frac{s}{2} \times 4\rho$,
  thus there are at most $\frac{1}{4\rho}$ rounded processing times 
  between $\frac{s}{2}$ and $s$.
\end{proof}

In the case of compressible jobs, the rounding step can be simplified:
\begin{lemma}
  \label{lemma:largest-height}
  Compressible jobs must have one of the two largest rounded processing times.
\end{lemma}

\begin{proof}
  Let $s$ be again the shelf height,
  and $j$ be a job that is wide in the shelf, i.e.~$\procnum{j}{s} \geq b$.
  Assume $\ptm{j}{\procnum{j}{s}} \leq \frac{1}{1+4\rho}s$.
  Then $\ptm{j}{\floor{(1-\rho)\procnum{j}{s}}} \leq s$ by Lemma~\ref{lemma:compression},
  contradicting the definition of~$\procnum{j}{s}$.
  Therefore, $\ptm{j}{\procnum{j}{s}} > \frac{1}{1+4\rho}s$.
  Since $\mathop{\mathrm{geom}}(\frac{s}{2}, s, 1 + 4\rho)$
  contains at most one value in the interval $(\frac{1}{1+4\rho}s, s]$,
  $\ptm{j}{\procnum{j}{s}}$ is either rounded down to that value, or the one below.
\end{proof}

Jobs that are wide in shelf~$S_1$ must also be wide in shelf~$S_2$,
so there are $\Landau{<=}{\frac{\log (\delta m)}{\delta}}$ different processor counts and two
processing times each for shelf~$S_1$ and shelf~$S_2$ to consider.
This leaves us with no more than $k_C = \Landau{<=}{\frac{1}{\delta^2} \log^2 (\delta m)}$~types of compressible items.

Now consider the types of incompressible jobs.
These jobs have $b-1 = \Landau{<=}{\frac{1}{\delta}}$ many different processor counts in shelf~$S_1$.
They can be narrow or wide in shelf~$S_2$.
If a job is wide in shelf~$S_2$,
we round the processing times according to Lemma~\ref{lemma:heights-bound}.
Then there are $\Landau{<=}{\frac{1}{\delta}}$ many different processing times in shelf~$S_1$,
and $\Landau{<=}{\frac{\log (\delta m)}{\delta}}$ different processor counts, but only two processing times possible in shelf~$S_2$.
Therefore, there are
at most~$k = \Landau{<=}{\frac{1}{\delta^3} \log (\delta m)}$~item types for such jobs.

If a job is narrow in shelf~$S_2$, we directly rounded the profit down to $0$ or
up to the next value in $\mathop{\mathrm{geom}}(\frac{\delta}{2}d, \frac{b}{2}d, 1 + \frac{\delta}{b})$.
This leaves us with $\Landau{<=}{\frac{1}{\delta^2} \log \frac{1}{\delta}}$~many different profits,
or $\Landau{<=}{\frac{1}{\delta^3} \log \frac{1}{\delta}}$~item types.

In total, there are at most $k_I = \Landau{<=}{\frac{1}{\delta^3} (\log \frac{1}{\delta} + \log (\delta m)} = \Landau{<=}{\frac{1}{\delta^3} \log m}$~types of incompressible items.

\subsubsection{Putting it Together}

Our algorithm, \cref{alg:scheduling-improved} is very similar to \cref{alg:scheduling-simple}.
The main difference is how we solve the knapsack instance in \cref{line:bounded-knapsack}.
\begin{algorithm}
  \Input{$J, m, d, \epsilon$}
  $\delta \leftarrow \frac{1}{5} \epsilon$, $\rho \leftarrow \frac{1}{4}(\sqrt{1+\delta} - 1)$,
    $b \leftarrow \frac{1}{2\rho - \rho^2}$, $d' \leftarrow (1+\delta)^2 d$,
    $J^C \leftarrow \setst{j \in J_B(d)}{\procnum{j}{d} \geq \frac{1}{\rho}}$ \;
  \For{$j \in J_B(d)$}{
    Precompute $\procnum{j}{\frac{d}{2}}, \procnum{j}{d}$,
      $\procnum[2]{j}{\frac{d'}{2}}$, $\procnum{j}{d'}$, $\procnum{j}{\frac{3}{2}d'}$,
      $\s{j}$, $\pr{j}$ \;
  }
  Find solution $J'$ to $(J_B(d), J^C, m, \rho)$ via bounded knapsack \label{line:bounded-knapsack}\;
  Apply Lemma~\ref{lemma:make-feasible} to $J' \cap J_B(d')$,
    obtain schedule for jobs $J$ with makespan $\frac{3}{2}d'$ \;
  \lIf{schedule is infeasible}{
    \KwSty{reject} $d$}
  \lElse{
    \Return{the schedule}}
  
  \caption{Scheduling of monotone moldable jobs using bounded knapsack with compressible items}
  \label{alg:scheduling-improved}
\end{algorithm}
We interpret it as an instance of the bounded knapsack problem,
transform it into a regular knapsack instance $(I, I_C, m, \rho)$
with $\Landau{<=}{k_C \log m}$~compressible and $\Landau{<=}{k_I \log m}$~incompressible items,
and solve it using \cref{alg:knapsack-compressible}.
We proceed to replace the container items in the solution
with the appropriate number of items (jobs) from the corresponding type.
Note that we used only part of the compressibility
for the solution of the knapsack problem in \cref{line:bounded-knapsack}.
The rest is required to compensate for the rounding of the processor counts.

\begin{lemma}
  \label{lemma:unrounded-feasible}
  When the original processor counts and processing times are considered,
  $J'$ is a solution of the knapsack problem $\KP{J_B(d)}{m, d'}$,
  and, unless there is no schedule with makespan~$d$,
  $W(J', d') \leq md' - W_S(d)$.
\end{lemma}

\begin{proof}
  We have to carefully consider the implications of our rounding of the jobs on~$J'$.
  We underestimated the processor count (size) of wide jobs by a factor of at most~$1 - \rho$.
  The total size of all selected containers does not exceed $m$
  when the size of containers for wide jobs is multiplied by a factor of $1 - \rho$.
  According to Lemma~\ref{lemma:compress-twice},
  one compression with factor $2\rho - \rho^2$ reduces the processor count
  by a factor of~$(1 - \rho)^2$,
  and increases the processing time by a factor of less than $1 + \delta$.
  The jobs then fit into $m$~machines.
  Even before the rounding, the processing time was at most~$d$,
  so after the compression they have processing time
  less than $(1 + \delta) d < d'$,
  so $J'$ is a solution to $\KP{J_B(d)}{m, d'}$.
  
  We further have to bound the work~$W(J', d')$,
  assuming there is a schedule with makespan $d$.
  By Lemma~\ref{lemma:two-shelf-schedule}
  there is a two-shelf-schedule with work at most~$md - W_S(d)$.
  This must also be true when the rounded processor counts and processing times are considered.

  But our solution to the knapsack problem may not be optimal,
  since we modified the profits of the items corresponding jobs that are narrow in both shelves.
  Because a solution can contain at most $m$ jobs,
  omitting jobs with profit below~$\frac{\delta}{2}d$ reduces the profit by at most~$\frac{\delta}{2} m d$.
  Rounding the profits of the other items up lead us to overestimate the profit of the solution
  by at most $m b \frac{d}{2} \frac{\delta}{b} = \frac{\delta}{2} m d$,
  since each of the at most~$m$ items has original profit at most~$b \frac{d}{2}$,
  which we increased by a factor of at most~$1 + \frac{\delta}{b}$.
  Since the profit equates to the saved work,
  the work of our schedule with rounding may be up to~$md - W_S(d) + \delta m d$.

  Undoing the rounding can also increase the work of our schedule,
  namely by a factor of $1 + 4\rho$ for the rounding of the processing times,
  and a factor of $1 + \rho$ for the rounding of the processor counts.
  So with the original, unrounded numbers we must have
  \begin{equation}
    \begin{aligned}
      W(J', d') &\leq (md - W_S(d) + \delta m d)(1+4\rho)(1 + \rho) \\
      &\leq (m (1+\delta) d - W_S(d)) (1+4\rho)^2 \\
      &\leq m d (1+\delta) (1+4\rho)^2 - W_S(d) \\
      &\leq md' - W_S(d) \dpunct{.}
    \end{aligned}
  \end{equation}
\end{proof}

According to Corollary~\ref{thm:approx},
\cref{alg:scheduling-improved} yields a schedule with makespan $\frac{3}{2}d' \leq (\frac{3}{2} + \epsilon) d$.

The running time of \cref{line:bounded-knapsack} includes the transformations of the instances
and the solution, which can be done in $\Landau{<=}{n + k_I + k_C}$.
For the solution of the knapsack instance itself,
Theorem~\ref{thm:knapsack-compressible} states that
time~$\Landau{<=}{n_I \beta_{\max} + n_C \bar{n} \frac{1}{\rho'} \log \frac{C}{\alpha_{\min}}}$ is required,
where $n_I = \Landau{<=}{k_I \log m}$ and $n_C = \Landau{<=}{k_C \log m}$.
Using the same bounds~$\alpha_{\min}$, $\beta_{\max}$, and~$\bar{n}$ as before,
the running time is~$\Landau{<=}{n \frac{1}{\epsilon^2} \log m \parens[2]{\frac{\log m}{\epsilon} + \log^3 (\epsilon m)}}$.
Again, this is the dominating step.

However, while the solution of the knapsack problem is now linear in the number~$n$ of jobs,
applying the transformation rules (see \cref{sec:transformation-rules}).
to create a feasible schedule still requires $\Landau{<=}{n \log n}$~operations.
The total running time thus is
$\Landau{<=}{n \frac{1}{\epsilon^2} \log m \parens[2]{\frac{\log m}{\epsilon} + \log^3 (\epsilon m)} + n \log n}$.

\subsubsection{Obtaining a Linear Algorithm}
\label{sec:linear}

The super-linear running time when applying the transformations
stems from organizing the jobs in $S_1$ in a heap.
Instead of using the exact processing time~$\ptm{j}{\procnum{j}{d}}$,
we can use the rounded processing time~$\check{t}_j(d)$ we introduced in the last section.
Then the jobs can be organized in~$\Landau{<=}{\frac{1}{\delta}}$ lists.
The running time for the transformation rules then is~$\Landau{<=}{\frac{n}{\delta}}$.
We could even organize the lists of jobs in a heap,
resulting in a running time of~$\Landau{<=}{n \log \frac{1}{\delta}}$.

Since we underestimate the processing time of the jobs by at most~$\delta d$,
the makespan of the final schedule is at most~$(\frac{3}{2}(1+\delta)^2 + \delta) d$.
Using $\delta = \frac{1}{5}\epsilon$ will result in a $(\frac{3}{2} + \epsilon)$-dual approximation algorithm.
Since the other steps of the dual algorithm remain the same,
the total running time is
$\Landau{<=}{n \frac{1}{\epsilon^2} \log m \parens[2]{\frac{\log m}{\epsilon} + \log^3 (\epsilon m)}}$.

\section*{Conclusion}

We have presented several techniques to exploit monotony of jobs,
mostly based on the ability to compress wide jobs.
For two different algorithms we demonstrated that these techniques can help to reduce the running time
from polynomial in~$m$ to polynomial in~$\log m$.
We also showed that arbitrarily good approximation guarantees can be achieved in polynomial time.
On the negative side, we proved the \cclass{NP}-hardness of scheduling monotonic jobs.
It remains open whether a better approximation guarantee than~$\frac{3}{2}$
can be achieved efficiently, e.g.~in the form of an EPTAS.

\bibliographystyle{plain}
\bibliography{bibliography}

\end{document}